\RequirePackage{amsmath}
\documentclass{llncs}
\usepackage{makeidx}  

\usepackage[numbers,sort&compress]{natbib}
\usepackage{semantic}
\usepackage{thmtools}
\usepackage{thm-restate}
\usepackage{etoolbox}
\usepackage{lipsum}
\usepackage{stmaryrd}
\usepackage{amssymb}
\usepackage{url}

\newcommand{\EXP}[1]{\llbracket #1 \rrbracket}
\newcommand{\CE}[0]{\mathcal{E}}
\newcommand{\SEM}[1]{\mathcal{E}\llbracket #1 \rrbracket}
\newcommand{\PSEM}[1]{\mathcal{E}_P[\![ #1 ]\!]}
\newcommand{\ESEM}[1]{\mathcal{E}_E[\![ #1 ]\!]}
\newcommand{\XSEM}[1]{\mathcal{E}_X[\![ #1 ]\!]}
\newcommand{\TSEM}[1]{\mathcal{T}\llbracket #1 \rrbracket}
\newcommand{\lam}[1]{\lambda #1.\,}
\newcommand{\of}[0]{{:}}
\newcommand{\by}[0]{{:=}}
\newcommand{\INT}[0]{\mathtt{int}}

\newcommand{\fix}[1]{\mathtt{fix}\,#1.\,}

\newcommand{\fst}[1]{\mathsf{fst}(#1)}
\newcommand{\snd}[1]{\mathsf{snd}(#1)}
\newcommand{\emptyenv}[0]{\mathsf{empty}}
\newcommand{\ext}[3]{\mathsf{extend}(#1,#2,#3)}
\newcommand{\tyext}[1]{\mathsf{tyExtend}(#1)}
\newcommand{\lookup}[2]{\mathsf{lookup}(#1,#2)}
\newcommand{\shift}[3]{\,\uparrow^{#1}_{#2}(#3)}
\newcommand{\app}[0]{\;}

\newcommand{\abs}[1]{\mathsf{thunk}(#1)}
\newcommand{\val}[1]{\overline{#1}}
\newcommand{\CASE}[0]{\textbf{case}\,}
\newcommand{\OF}[0]{\,\textbf{of}\,}
\newcommand{\IF}[0]{\textbf{if}\,}
\newcommand{\THEN}[0]{\,\textbf{then}\,}
\newcommand{\ELSE}[0]{\,\textbf{else}\,}
\newcommand{\return}[0]{\mathsf{return}\;}

\newcommand{\down}[0]{\mathsf{down}\;}
\newcommand{\wrong}[0]{\mathsf{wrong}}

\newcommand{\zero}[0]{\mathsf{zero}}
\newcommand{\SET}[1]{\mathcal{P}(#1)}
\newcommand{\FSET}[1]{\mathcal{P}_f(#1)}

\begin{document}
\frontmatter          
\pagestyle{headings}  
\addtocmark{Revisiting Elementary Semantics} 
%
%
%
%
\title{Revisiting Elementary Denotational Semantics}
\titlerunning{Revisiting Elementary Semantics}  
%
\author{Jeremy G. Siek}
\authorrunning{Jeremy G. Siek} 
%
\tocauthor{Jeremy G. Siek}
\institute{Indiana University, Bloomington IN 47401, USA,\\
\email{jsiek@indiana.edu}}

\maketitle              

\begin{abstract}
Operational semantics have been enormously successful, in large part
due to its flexibility and simplicity, but they are not compositional.
Denotational semantics, on the other hand, are compositional but the
lattice-theoretic models are complex and difficult to scale to large
languages.  However, there are elementary models of the
$\lambda$-calculus that are much less complex: by Coppo,
Dezani-Ciancaglini, and Salle (1979), Engeler (1981), and Plotkin
(1993).

\vspace{4pt}

This paper takes first steps toward answering the question: can
elementary models be good for the day-to-day work of language
specification, mechanization, and compiler correctness?  The
elementary models in the literature are simple, but they are not as
intuitive as they could be. To remedy this, we create a new model that
represents functions literally as finite graphs. Regarding
mechanization, we give the first machine-checked proof of soundness
and completeness of an elementary model with respect to an operational
semantics. Regarding compiler correctness, we define a polyvariant
inliner for the call-by-value $\lambda$-calculus and prove that its
output is contextually equivalent to its input.  Toward scaling
elementary models to larger languages, we formulate our semantics in a
monadic style, give a semantics for System F with general recursion,
and mechanize the proof of type soundness.

\end{abstract}

\section{Introduction}
\label{sec:intro}


This paper revisits elementary models of the
$\lambda$-calculus~\citep{Engeler:1981aa,Plotkin:1993ab,Coppo:1979aa,Barendregt:2013aa}
with an eye towards determining whether they are a suitable choice for
modern programming language semantics. That is, are they good for the
day-to-day work of language specification, mechanization, and compiler
correctness?  The author hypothesizes that the answer is yes because
these models satisfy three important properties.
\begin{description}
\item[compositional] The semantics are defined by structural recursion
  on syntax~\citep{Winskel:1993uq}.
\item[extensional] The semantics specify externally observable
  behavior~\citep{Schmidt:2003aa}.
\item[elementary] The semantics use quite simple mathematics~\citep{Plotkin:1993ab}.
\end{description}

Compositionality enables proof by structural induction on the syntax,
which simplifies proofs of properties such as type soundness and
compiler correctness.  Extensionality is beneficial because,
ultimately, a language specification must be extensional; intensional
semantics require the circuitous step of erasing internal
behavior. The use of only elementary mathematics, that is, mathematics
familiar to undergraduates in computer science, is beneficial because
language specifications should be readable by all computer scientists
and because the size of a language's metatheory depends on the
complexity of the mathematics.


Historically, the two main approaches to specifying programming
languages have been denotational
semantics~\citep{Scott:1970dp,Scott:1971ab} and operational
semantics~\citep{Church:1932aa,McCarthy:1960dz,Landin:1964dk,Plotkin:1981aa,Felleisen:1987ab}.
Denotational semantics are compositional and extensional but the
standard lattice-based models are not elementary.  Their mathematical
complexity is evident in the size of mechanized definitions of the
$\lambda$-calculus: \citet{Benton:2009ab} build a model in 11,000 LOC
and \citet{Dockins:2014aa} in 54,104 LOC (though it is difficult to
determine how many of those LOC is strictly necessary).

Operational semantics are elementary but neither extensional nor
compositional. A mechanized definition of the $\lambda$-calculus using
operational techniques is under 100 LOC.
Small-step semantics are intensional in that the input-output behavior
of a program is a by-product of a sequence of transitions. Big-step
semantics are intensional in that the value of a $\lambda$ abstraction
is a syntactic object, a closure~\citep{Kahn:1987aa}.
The lack of compositionality in operational semantics imposes
significant costs when reasoning about programs: sophisticated
techniques such as logical
relations~\citep{Plotkin:1973fk,Statman:1985aa} and
simulations~\citep{Aczel:1988aa,Milner:1995aa} are often
necessary. For example, the correctness proofs for the CompCert C
compiler made extensive use of simulations but sometimes resorted to
translation validation in cases where verification was too difficult
or expensive~\citep{Tristan:2008aa,Leroy:2009aa}. Likewise, the
logical relations necessary to handle modern languages are daunting in
their complexity~\citep{Ahmed:2009aa,Hur:2011aa}.

But what if there were denotational semantics that were also
elementary? In fact, in the 1970's and 1980's several groups of
researchers discovered elementary models for the untyped
$\lambda$-calculus.  \citet{Plotkin:1972aa,Plotkin:1993ab} and
\citet{Engeler:1981aa} discovered elementary models based on two
insights:
\begin{enumerate}
\item it suffices to use a finite approximation of a function's graph
  when passing it to another function, and
\item self application can be handled by allowing larger
  approximations to be used when a smaller approximation is expected.
\end{enumerate}
\citet{Coppo:1979aa} discovered type-theoretic models for the
$\lambda$-calculus based on two insights:
\begin{enumerate}
\item that the behavior of $\lambda$ abstractions can be completely
  characterized by intersection types, and
\item self application can be handled using subtyping.
\end{enumerate}
\citet{Plotkin:1993ab} shows how the elementary and type-theoretic
models are closely related. The above pairs of insights are really the
same insights.

\paragraph{Contributions} This paper makes several technical
contributions that begin to answer the question of whether elementary
models are a good choice for programming language specification,
mechanization, and compiler correctness. All of the results in this
paper are mechanized in Isabelle and are available in the Archive of
Formal Proof~\citep{Siek:2017aa}.

\begin{enumerate}
\item To yield a more intuitive model for the CBV $\lambda$-calculus,
  we construct a domain that approximates functions literally by their
  finite graphs (Section~\ref{sec:denot-function}). The mechanization
  is under 100 LOC.

\item We give a type-theoretic version of this model based on
  intersection types (50 LOC) and prove that the two are isomorphic
  (Section~\ref{sec:denot-type-system}).

\item We give the first mechanized proofs of soundness and
  completeness for a elementary model with respect to operational
  semantics. We also mechanize soundness with respect to contextual
  equivalence (Section~\ref{sec:correct}).

\item We show how compositionality can be beneficial by proving
  correctness for a compiler optimization pass that performs inlining
  in under 100 LOC (Section~\ref{sec:optimize}).

\item Toward scaling to more language features, we formulate the
  semantics in a monadic style
  (Section~\ref{sec:non-determinism-monad}) and we define a semantics
  for a language with first-class parametric polymorphism and general
  recursion and mechanize its proof of semantic type soundness
  (Section~\ref{sec:cannot-go-wrong}).

\end{enumerate}

We begin with a review of the three elementary semantics of
\citet{Plotkin:1993ab}, \citet{Engeler:1981aa}, and
\citet{Coppo:1979aa} (Section~\ref{sec:elementary-semantics}).  We
discuss related work in Section~\ref{sec:related-work} and 
conclude in Section~\ref{sec:conclusion}.

\section{Background on Elementary Semantics}
\label{sec:elementary-semantics}

\begin{quote}
  \emph{But $f$ is a \emph{function}; an infinite object.  What does it mean
  to ``compute'' with an ``finite'' argument?  In this case it means
  most simply that $h(f)$ is determined by asking of $f$ finitely many
  questions: $f(m_0), f(m_1), ..., f(m_{k-1})$.} ---\citet{Scott:1993uq}
\end{quote}

\begin{figure}[tbp]
\[
\begin{array}{rcll}
  x & \in & \mathbb{X} & \text{variables} \\
  n & \in & \mathbb{Z} & \text{integers} \\
  \oplus & \in & \{ +, \times, -, \ldots \} & \text{arithmetic operators}\\
  e \in \mathbb{E} & ::= & n \mid e \oplus e \mid x \mid \lam{x} e \mid e \; e
    \mid \IF e \THEN e \ELSE e \quad
    & \text{expressions}
\end{array}
\]
\caption{Syntax of a call-by-value $\lambda$ calculus with integer arithmetic.}
  \label{fig:lambda-calculus}
\end{figure}

We review elementary semantics in the setting of a call-by-value (CBV)
untyped $\lambda$-calculus extended with integer arithmetic. The
syntax of this $\lambda$-calculus is defined in
Figure~\ref{fig:lambda-calculus}. We write $n$ for integers, $e_1
\oplus e_2$ for arithmetic operations, $x$ for variables, $\lam{x} e$
for abstraction, $e_1\app e_2$ for application, and $\IF e_1 \THEN e_2
\ELSE e_3$ for conditionals.

\subsection{Set-Theoretic Models}
\label{sec:set-theoretic-models}

The domains of \citet{Plotkin:1972aa,Plotkin:1993ab} and
\citet{Engeler:1981aa} are $\SET{\mathbb{D}_P}$ and
$\SET{\mathbb{D}_e}$, where $\mathbb{D}_P$ and $\mathbb{D}_E$ and
inductively defined by the following recursive equations.
\[
  \mathbb{D}_P = \mathbb{Z} + \FSET{\mathbb{D}_P} \times \FSET{\mathbb{D}_P}
  \qquad
  \mathbb{D}_E = \mathbb{Z} + \FSET{\mathbb{D}_E} \times \mathbb{D}_E
\]
We let $d$ range over elements of $\mathbb{D}_P$ or $\mathbb{D}_E$,
and $D$ ranges over finite sets of them. For Plotkin, an element
$(D,D') \in \SET{\mathbb{D}_P}$ represents a single input-output entry
in the graph of a function. For first-order functions over integers,
$D$ and $D'$ are just singletons. For higher-order functions, $D$
and $D'$ are finite subsets of a function's graph. It turns out that
finite sets in output position are not necessary.  One of Plotkin's
entries $(D, \{ d'_1,\ldots, d'_n\} )$ is instead represented by
multiple entries $(D, d'_1), \ldots, (D, d'_n)$ in Engeler's model.

The wonderful thing about $\SET{\mathbb{D}_P}$ and
$\SET{\mathbb{D}_e}$ is their simplicity. Their construction does not
require advanced techniques such as inverse
limits~\citep{Scott:1970dp,Wand:1979aa,Smyth:1982aa}. Both
$\mathbb{D}_P$ and $\mathbb{D}_E$ are straightforward to define as
algebraic datatypes in proof assistants such as Isabelle or Coq.  In
Isabelle, the \texttt{fset} library provides finite sets, but one
could also use lists at the cost of a few extra lemmas.

\begin{figure}[tbp]
\begin{align*}
\PSEM{ \lam{x} e }\rho &= 
  \{ (D,D') \mid D' \subseteq \PSEM{e}\rho(x{:=}D)\} \\
\PSEM{e_1\app e_2}\rho &= \bigcup \{ D' \mid \exists D.\,
  (D,D') \in \PSEM{e_1}\rho \land D \subseteq \PSEM{e_2}\rho \} \\[1ex]
\ESEM{ \lam{x} e }\rho &= 
  \{ (D,d') \mid d' \in \ESEM{e}\rho(x{:=}D)\} \\
\ESEM{e_1\app e_2}\rho &= \{ d' \mid \exists D.\,
  (D,d') \in \ESEM{e_1}\rho \land D \subseteq \ESEM{e_2}\rho \} \\[1ex]
\XSEM{x}\rho &= \rho(x) \\
\XSEM{ n }\rho &= \{ n \} \\
\XSEM{ e_1 \oplus e_2 }\rho &= \{  n_1 \oplus n_2 \mid 
   n_1 \in \XSEM{ e_1 }\rho \land n_2 \in \XSEM{ e_2 }\rho \} \\
\XSEM{\IF e_1 \THEN e_2 \ELSE e_3}\rho &=
\left\{ v\, \middle|\,
    \exists n.\, n \in \XSEM{e_1}\rho 
  \begin{array}{l}
   \land\, (n\neq 0 \implies v \in \XSEM{e_2}\rho)\\
   \land\, (n=0 \implies v \in \XSEM{e_3}\rho)
  \end{array}
  \right\}
\end{align*}
\caption{Two elementary semantics for CBV $\lambda$-calculus,
  $\mathcal{E}_P$ using Plotkin's model and $\mathcal{E}_E$ using
  Engeler's. The common parts are parameterized, i.e., $\mathcal{E}_X$
  where $X \in \{P,E\}$.}
\label{fig:plotkin-engeler}
\end{figure}

\citet{Plotkin:1972aa,Plotkin:1993ab} used his domain to give a
semantics to the $\lambda\beta$-calculus whereas
\citet{Engeler:1981aa} used his to give semantics to combinatory logic
(the S and K combinators). In this paper, we are instead concerned
with a CBV $\lambda$-calculus. To make for a clear comparison with our
work, we adapt their semantics to CBV $\lambda$-calculus, defining
$\mathcal{E}_P$ and $\mathcal{E}_E$ in
Figure~\ref{fig:plotkin-engeler}. We conjecture that these two
semantics are equivalent to our own.

Now to explain the two semantics in
Figure~\ref{fig:plotkin-engeler}. 
As usual, we write
$\rho(x{:=}d)$ for the map that sends $x$ to $d$ and any other
variable $y$ to $\rho(y)$.  In $\mathcal{E}_P$, the meaning of an
abstraction $\lam{x} e$ is the set of all input-output entries
$(D,D')$ such that $D'$ is a subset of the meaning of the body $e$ in
a context where $x{:=}D$. Similarly, in $\mathcal{E}_E$, the meaning
of an abstraction $\lam{x} e$ is the set of all input-output entries
$(D,d')$ such that $d'$ is an element in meaning of the $e$ with
$x{:=}D$. Regarding the meaning of function application $(e_1\,e_2)$,
$\mathcal{E}_E$ collects up all of the outputs $d'$ from the entries
$(D,d')$ in the meaning of $e_1$ whenever $D$ an finite approximation
of the argument $e_2$.  A finite approximation is good enough because,
during a terminating call to a higher-order function, only a finite
number of calls will be made to its argument. For a non-terminating
call, the semantics assigns the meaning $\emptyset$. The use of subset
in $D \subseteq \ESEM{e_2}\rho$ is critically important, as it enables
self application and thereby general recursion via the $Y$
combinator\footnote{Really the $Z$ combinator because we are in a
  call-by-value setting}.  The $\mathcal{E}_P$ semantics for function
application is slightly more complex because the outputs $D'$ must be
flattened to produce something in $\SET{\mathbb{D}_P}$ and not
$\SET{\FSET{\mathbb{D}_P}}$.  Note that in both of these semantics,
the environment $\rho$ maps variables to finite sets, which either
encode an integer $n$ with a singleton $\{n\}$ or encode a finite
approximation of a function $\{(D_1,D'),\ldots,(D_n,D'_n)\}$.

\subsection{Type-Theoretic Models}
\label{sec:filter-models}

\citet{Coppo:1979aa} showed that a type system based on intersection
types can characterize the behavior of $\lambda$ terms, in particular,
showing that their type system induced the same equalities as the
$\mathcal{P}(\omega)$ model of \citet{Scott:1976lq}. Their work led to
a long line of research on filter models based on intersection types
for many different
$\lambda$-calculi~\citep{Hindley:1982aa,Coppo:1984aa,Coppo:1987aa,Alessi:2003aa,Alessi:2006aa}. \citet{Barendregt:2013aa}
give a detailed survey of this work.  Here we review an intersection
type system that characterizes CBV
$\lambda$-calculus~\citep{Coppo:1984aa,Rocca:2004aa,Alessi:2006aa}.

Figure~\ref{fig:intersection-type-system} defines the syntax for types
$A,B,C \in \mathbb{T}$, a subtyping relation $A<:B$, and a type system
$\Gamma \vdash e : A$.  We write $\top$ for the ``top'' of all
function types, written $\nu$ in the literature~\citep{Egidi:1992aa}.
The $\lambda$-calculus we study here includes integers and arithmetic,
so we have added a singleton type $n$ for every integer. We define a
filter model $\mathcal{E}_C$ in terms of the type system as
follows. The domain is $\mathcal{P}(\mathbb{T})$.
\[
   \mathcal{E}_C \llbracket e \rrbracket \Gamma =
     \{ A \mid \Gamma \vdash e : A \}
\]
The name \emph{filter} comes from topology and order theory, and
refers to a set that is upward closed and closed under finite
intersection. These two properties are satisfied by fiat in
intersection type systems because of the subsumption and
$\wedge$-introduction rules.

\begin{figure}[tbp]
  Types
  \[
   A,B,C \in \mathbb{T} ::= n \mid A \to B \mid A \wedge B \mid \top
  \]
  Subtyping \hfill\fbox{$A <: B$}
\begin{gather*}
  A \to B <: \top
  \quad
  A <: A
  \quad
  A <: A \wedge A
  \quad
  A \wedge B <: A
  \quad
  A \wedge B <: B
  \\[1ex]
  (A \to B) \wedge (A \to C) <: A \to (B \wedge C)
  \\[1ex]
  \inference{A <: B & B <: C}{A <: C}
  \quad
  \inference{A <: A' & B <: B'}{A \wedge B <: A' \wedge B'}
  \quad
  \inference{A' <: A & B <: B'}{A \to B <: A' \to B'}
\end{gather*}
Typing \hfill\fbox{$\Gamma \vdash e : A$}
\begin{gather*}
\inference{}
          {\Gamma \vdash n : n}
\quad
\inference{\Gamma \vdash e_1 : n_1 & \Gamma \vdash e_2 : n_2}
          {\Gamma \vdash e_1 \oplus e_2 : n_1 \oplus n_2}
\quad
\inference{}{\Gamma \vdash x : \Gamma(x)}
\\[1ex]
\inference{\Gamma,x{:}A \vdash e : B}{\Gamma \vdash \lam{x} e : A \to B}
\quad
\inference{}{\Gamma \vdash \lam{x} e : \top}
\quad
\inference{\Gamma \vdash M : A & \Gamma \vdash M : B}
          {\Gamma \vdash M : A \wedge B}
\\[1ex]
\inference{\Gamma \vdash e_1 : A \to B &
           \Gamma \vdash e_2 : A}
          {\Gamma \vdash e_1 \app e_2 : B}
\quad
\inference{\Gamma \vdash e : A & A <: B}
          {\Gamma \vdash e : B}
\\[1ex]
\inference{\Gamma \vdash e_1 : n & n \neq 0 & \Gamma \vdash e_2 : B}
          {\Gamma \vdash \IF e_1 \THEN e_2 \ELSE e_3 : B}
\quad
\inference{\Gamma \vdash e_1 : n & n = 0 & \Gamma \vdash e_3 : B}
          {\Gamma \vdash \IF e_1 \THEN e_2 \ELSE e_3 : B}
\end{gather*}
\caption{An intersection type system that characterizes
  the CBV $\lambda$-calculus.}
\label{fig:intersection-type-system}
\end{figure}

\citet{Alessi:2003aa} show, for many variations of intersection type
systems, that typing is preserved by both reduction and expansion,
that is
\begin{description}
\item[Preservation under Reduction] If $\Gamma \vdash e : A$ and $e \longrightarrow e'$,
  then $\Gamma \vdash e' : A$.
\item[Preservation under Expansion] If $\Gamma \vdash e' : A$ and $e \longrightarrow e'$, 
  then $\Gamma \vdash e : A$.
\end{description}
While type systems generally preserve types under reduction,
preserving under expansion is unusual and is what enables intersection
type systems to completely characterize the behavior of a program.  In
terms of the filter model, meaning is invariant under reduction: $e
\longrightarrow e'$ implies $\mathcal{E}_C \llbracket e \rrbracket =
\mathcal{E}_C \llbracket e' \rrbracket$.

\section{A Straightforward Elementary Semantics}
\label{sec:denot-function}

We wish to use elementary semantics for the specification of real
programming languages, so it is important that the semantics be as
intuitive as possible. To pick some nits, the placement of finite sets
in the domains of Plotkin and Engeler is unintuitive. For example, one
might naively think that the meaning of $(\lam{x}x+1)$ would include
an input-output entry such as $(3,4)$.  But in Plotkin's model we
instead have $(\{3\},\{4\})$ and in Engeler's model we have
$(\{3\},4)$. However, there is a domain from the semantic subtyping
literature that places the finite sets where one would
expect~\citep{Frisch:2008aa}.
\[
  \mathbb{D} = \mathbb{Z} + \FSET{\mathbb{D} \times \mathbb{D}}
\]
The idea is straightforward: a function is represented by a finite
approximation of its graph.  To our knowledge, this domain has never
been used to give meaning to programs, only to types.

So the domain for our elementary semantics is $\SET{\mathbb{D}}$.  Let
$t$ (for table) range over $\FSET{\mathbb{D} \times \mathbb{D}}$.  We
define the semantics $\CE$ in Figure~\ref{fig:denot-function} to take
an expression and an environment and return a set of elements.  Thanks
to the change in domain, the \emph{environment} $\rho$ is simply a
partial map from variables to elements (not sets of elements).

\begin{figure}[tbp]
\hfill\fbox{$\SEM{e}\rho$}
\begin{align*}
\SEM{ \lam{x} e }\rho &= 
  \{ t \mid \forall (d,d')\in t.\, d' \in \SEM{ e }\rho(x{:=}d) \} \\
\SEM{ e_1\;e_2 }\rho &= \left\{ d \, \middle| 
   \begin{array}{l}
   \exists t d_1 d_1' d_2 .\, t {\in} \SEM{ e_1 }\rho \land d_2 {\in} \SEM{ e_2 }\rho \\
   \land\, (d_1, d_1') \in t \land d_1 \sqsubseteq d_2 
   \land d \sqsubseteq d_1'
   \end{array}
\right\} \\
\SEM{ x }\rho &= \{ d \mid d \sqsubseteq \rho(x) \} \\
\SEM{ n }\rho &= \{ n \} \\
\SEM{ e_1 \oplus e_2 }\rho &= \{  n_1 \oplus n_2 \mid 
   n_1 \in \SEM{ e_1 }\rho \land n_2 \in \SEM{ e_2 }\rho \} \\
\SEM{\IF e_1 \THEN e_2 \ELSE e_3}\rho &=
\left\{ d\, \middle|\,
    \exists n.\, n \in \SEM{e_1}\rho 
  \begin{array}{l}
   \land\, (n\neq 0 \implies d \in \SEM{e_2}\rho)\\
   \land\, (n=0 \implies d \in \SEM{e_3}\rho)
  \end{array}
  \right\}
\end{align*}
\hfill\fbox{$d \sqsubseteq d$}\\[-4ex]
\[
    n \sqsubseteq n
    \qquad
    \inference{t \subseteq t'}{t \sqsubseteq t'}
\]
\vspace{-15pt}
\caption{A new elementary semantics for CBV $\lambda$-calculus.}
\label{fig:denot-function}
\end{figure}

The meaning of a $\lam{x}e$ is the set of all finite graphs $t$ such
that for every input-output entry $(d,d') \in t$, the output element
$d'$ is in the meaning of the body $e$ in the environment where
parameter $x$ is bound to the input element $d$.  The meaning of an
application $(e_1\app e_2)$ is basically given by table lookup.  That
is, if $t \in \SEM{e_1}\rho$, $(d_1,d'_1) \in t$, $d_2 \in
\SEM{e_2}\rho$ and $d_1 = d_2$, then $d'_1 \in \SEM{(e_1\app
  e_2)}\rho$. However, following in the footsteps of the prior
elementary semantics, we accommodate self application by allowing the
argument $d_2$ to be a larger approximation than the input entry
$d_1$. To this end we define an ordering relation $\sqsubseteq$ on
$\mathbb{D}$ in Figure~\ref{fig:denot-function}.  Thus, instead of
$d_1 = d_2$ we require $d_1 \sqsubseteq d_2$.

\subsection{Is this semantics a filter model?}
\label{sec:filter}

We don't yet know, but we have part of the story. Recall that a filter
is a set that is upward closed and closed under finite
intersection. First, our ordering $\sqsubseteq$ goes the opposite way,
so the question is whether the semantics forms an \emph{ideal} (the
dual of filter). Our semantics is downwards closed but we do not know
whether it is closed under finite union for arbitrary programs.
However, it is closed under finite unions for closed syntactic values
(integers and closed $\lambda$ abstractions) which is all that is
required for correctness in a CBV setting (Section~\ref{sec:correct}).

The following definition lifts the information ordering to
environments.
\begin{align*}
X \vdash \rho \sqsubseteq \rho' &\equiv
\forall x.\; x \in X \implies \rho(x) \sqsubseteq \rho'(x)
 \\
 \rho \sqsubseteq \rho' &\equiv
     \forall x.\; \rho(x) \sqsubseteq \rho'(x)
\end{align*}

\begin{proposition}[Downward closed aka. Subsumption]\ 
\begin{enumerate}
\item If $d \in \SEM{e}\rho$, $d' \sqsubseteq d$, 
   and $\text{fv}(e) \vdash \rho \sqsubseteq \rho'$,
   then $d' \in \SEM{e}\rho'$.
\item If $d \in \SEM{e}\rho$, $d' \sqsubseteq d$, 
      then $d' \in \SEM{e}\rho$.
\end{enumerate}
\label{prop:change-env-le}
\end{proposition}
\begin{proof}
  The proof of part (1) is by induction on $e$.
  Part (2) follows from part (1).
\end{proof}

The proof of Proposition~\ref{prop:change-env-le} is interesting in
that it influenced our definition of the semantics.  Regarding
variables, the semantics for a variable $x$ includes not just the
element $\rho(x)$ but also all elements below $\rho(x)$ (see
Figure~\ref{fig:denot-function}). If instead we had defined
$\SEM{x}\rho = \rho(x)$, then the case for variables in the proof of
Proposition~\ref{prop:change-env-le} would break. The same can be said
regarding $d \sqsubseteq d'_1$ in function application.

Next we consider finite unions, that is, we define a join operator on
$\mathbb{D}$.
\[
  n \sqcup n = n \qquad\qquad
  t \sqcup t' = t \cup t'
  \hspace{1in}\fbox{$d \sqcup d$}
\]
Of course the join is not always defined, e.g., there is no join of
the integers $0$ and $1$. The join is indeed the least upper bound of
$\sqsubseteq$.

\begin{proposition}[Join is the least upper bound of $\sqsubseteq$]\ 
\begin{enumerate}
\item $v_1 \sqsubseteq v_1 \sqcup v_2$, $v_2 \sqsubseteq v_1 \sqcup v_2$, and
\item If $v_1 \sqsubseteq v_3$ and $v_2 \sqsubseteq v_3$,
   then $v_1 \sqcup v_2 \sqsubseteq v_3$.
\end{enumerate}
\end{proposition}

We write $\val{v}$ for \emph{syntactic values} that are closed:
\[
\val{v} ::= n \mid \lam{x} e
    \qquad \text{where}\; \text{fv}(\lam{x}e) = \emptyset
\]
The following lemma says that, if we have two elements $d_1,d_2$ in
the meaning of a syntactic value $\val{v}$, then their join is too.

\begin{lemma}[Closed Under Join on Values]
  \label{lem:combine-values}
  If $d_1 \in \SEM{\val{v}}\rho$ and $d_2 \in \SEM{\overline{v}}\rho$,
  then $d_1 \sqcup d_2 \in \SEM{\val{v}}\rho$.
\end{lemma}

\subsection{Is this semantics fully abstract?}
\label{sec:full-abstraction}

No. Recall that a denotational semantics is \emph{fully abstract} if
contextual equivalence implies denotational equivalence. Many
denotational semantics are not fully abstract, and neither is our
elementary semantics. The upshot is that one cannot use the inequality
of two programs denotations to prove that they actually behave
differently. The counter example is the usual one: parallel-or.  The
parallel-or function, $\mathit{por}$, takes two thunks, and returns
true if either of them returns true.  It cannot be implemented in a
sequential language because inside $\mathit{por}$, either $f$ or $g$
must be called first. If the first thunk goes into an infinite loop,
then the second one will never be evaluated.

However, the domain $\SET{\mathbb{D}}$ contains a semantics for
$\mathit{por}$ as follows. For clarity, we add $unit$ and pairs to the
language and encode true as $1$ and false as $0$.
\[
\mathit{POR} =
\begin{array}{l}
  \{ ((t_1, t_2), 1) \mid  (unit, 1) \in t_1 \lor (unit, 1) \in t_2\} \\
  \cup 
  \{ ((t_1, t_2), 0) \mid (unit, 0) \in t_1 \land (unit, 0) \in t_2\}
\end{array}
\]
This denotation will return true if either argument to $\mathit{POR}$
returns true.

Now to use $\mathit{POR}$ to show that contextual equivalence does not
imply denotational equivalence. We first define two test functions
$T_i$ for $i \in \{ 0,1 \}$ as follows. Let $\Omega$ be the divergent
combinator $(\lam{x} x x) \app (\lam{x} x x)$.
\[
T_i = \lam{f}
\begin{array}[t]{l}
\IF f \app (\lam{x}1, \lam{x}\Omega) \THEN \\
\quad
\begin{array}{l}
  \IF f \app (\lam{x}\Omega, \lam{x}1) \THEN \\
    \quad \IF f \app (\lam{x}0, \lam{x}0) \THEN
      \Omega
    \ELSE i \\
  \ELSE \Omega
\end{array}
\\
\ELSE \Omega
\end{array}
\]
We have that $T_0 \simeq T_1$, but $\SEM{T_0} \neq \SEM{T_1}$. To see
why $T_0 \simeq T_1$, consider the possibilities for the input. They
could be given a function that 1) always diverges, 2) forces the first
thunk, 3) forces the second thunk, or 4) forces neither and always
returns the same thing which could be a) zero, b) a non-zero integer,
or c) a function. In case 1), both $T_0$ and $T_1$ diverge in the
first call to $f$. In case 2), depending on the result of the first
call to $f$, both $T_0$ and $T_1$ either get stuck, take the $\ELSE$
branch of the first $\IF$ and diverge, or take the $\THEN$ branch and
diverge in the second call to $f$. In case 3), both $T_0$ and $T_1$
diverge on the first call to $f$. In case 4a) both $T_0$ and $T_1$
take the $\ELSE$ branch of the first $\IF$ and diverge.  In case 4b)
both $T_0$ and $T_1$ take the $\THEN$ branch of all three $\IF$'s and
diverge. In case 4c), both $T_0$ and $T_1$ get stuck after the first
call to $f$.
So we have $T_0 \simeq T_1$, but it remains to show $\SEM{T_0} \neq
\SEM{T_1}$. We have that $(\mathit{POR}, 0) \in \SEM{T_0}$ but
$(\mathit{POR}, 0) \not\in \SEM{T_1}$.

We note that \citet{Rocca:2004aa} proved full abstraction for a filter
model for the CBV $\lambda$-calculus at the cost of some added
complexity.




\subsection{Example of a Recursive Function}
\label{sec:factorial}

To provide a concrete example of our semantics, we show the semantics
of the factorial function, implemented using the $Z$ combinator (the
$Y$ combinator for strict languages).  The main idea for how our
semantics gives meaning to recursive functions is that of a Matryoshka
doll: it nests ever-smaller versions of it's graph inside of them.

Recall the $Z$ combinator:
\[
  M  \equiv \lam{x} f \app (\lam{v} (x\app x) \app v) \qquad
  Z  \equiv \lam{f} M \app M 
\]
The factorial function is defined as follows, with a parameter $r$ for
calling itself recursively. We give names ($F$ and $H$) to the
$\lambda$ abstractions because we shall define tables for each of
them.
\[
  F  \equiv \lam{n} \IF n=0 \THEN 1\ELSE n \times r \app (n-1) \qquad
  H \equiv \lam{r} F \qquad
  \mathit{fact} \equiv Z\app H
\]
We begin with the tables for $F$: we define a function $F_t$ that
gives the table for just one input $n$; it simply maps $n$ to $n$
factorial.
\[
  F_t(n) = \{ (n,n!) \}
\]
The tables for $H$ map a factorial table for $n-1$ to a table for $n$.
We invite the reader to check that
$H_t(n) \in \SEM{H}\emptyset$ for any $n$.
\[
H_t(n) = \{ (\emptyset, F_t(0)), (F_t(0), F_t(1)), \ldots, (F_t(n-1), F_t(n))\}
\]
Next we come to the most important part: describing the tables for
$M$. Recall that $M$ is applied to itself; this is where the analogy
to Matryoshka dolls comes in.  We define $M_t$ by recursion on
$n$. Each $M_t(n)$ extends the smaller version of itself, $M_t(n-1)$,
with one more entry that maps the smaller version to the table for
factorial of $n-1$.
\begin{align*}
  M_t(0) &= \emptyset \\
  M_t(n) &= M_t(n-1) \cup \{ (M_t(n-1), F_t(n-1)) \}
\end{align*}
The tables $M_t$ are meanings for $M$ because $M_t(n+1) \in
\SEM{M}(f{:=}H_t(k))$ for any $n \leq k$.  The application of $M$ to
itself is OK, $F_t(n) \in \SEM{M\app M}(f{:=}H_t(k))$, because 
$(M_t(n),F_t(n)) \in M_t(n+1)$, $M_t(n)\sqsubseteq M_t(n+1)$, and
$F_t(n) \sqsubseteq F_t(n)$.

To finish things up, the tables for $Z$ map the $H_t$'s to the
factorial tables.
\[
  Z_t(n) = \{ (H_t(n), F_t(n)) \}
\]
Then we have $Z_t(n) \in \SEM{Z}\emptyset$ for all $n$.  So $F_t(n)
\in \SEM{Z\app H}\emptyset$ for any $n$ so we conclude that $n! \in
\SEM{\mathit{fact}\,n}\emptyset$.

\section{The Elementary Semantics as a Type System}
\label{sec:denot-type-system}

In this section we present a type system based on intersection types,
similar to the one in Section~\ref{sec:filter-models}, and mechanize
the proof that it is isomorphic to our elementary semantics.
%
%
The grammar of types is the following and consists of two
non-terminals, one for types of functions and one for types in
general.
\begin{align*}
  F,G,H  ::= & A \to B \mid F \wedge G \mid \top & \text{types of functions} \\
  A,B,C \in \mathbb{T} ::=& n \mid F   & \text{types}
\end{align*}
Again $n$ is a  singleton integer type.  The intersection
type $G \wedge H$ is for functions that have type $F$ and $G$.  All
functions have type $\top$.  The use of two non-terminals enables the
restriction of the intersection types to not include singleton integer
types, which would either be trivial (e.g. $1 \wedge 1$) or garbage
(e.g. $0 \wedge 1$).

We define subtyping $<:$ and type equivalence $\approx$ in
Figure~\ref{fig:subtyping} but with fewer rules than in
Figure~\ref{fig:intersection-type-system}.  We omit the following
rules, i.e., the subtyping rule for functions and for distributing
intersections through functions, because they were not needed for the
isomorphism with our elementary semantics, and therefore not needed in
the proofs of correctness (Section~\ref{sec:correct}).
\[
\inference{A' <: A & B <: B'}
          {A \to B <: A' \to B'}
\qquad
\inference{}{(C\to A) \wedge (C \to B) <: C \to (A \wedge B)}
\]

\begin{figure}[tbp]
\hfill\fbox{$A <: B$}\;\fbox{$F <: G$}
\begin{gather*}
  \inference{}{A <: A}
  \quad
  \inference{A <: B & B <: C}{A <: C}
  \quad
  \inference{A \approx A' & B \approx B'}{A \to B <: A' \to B'}
   \\[2ex]
    \inference{H <: F & H <: G}
              {H <: F \wedge G}
    \quad
    \inference{}
              {F \wedge G <: F}
    \quad
    \inference{}
              {F \wedge G <: G}
   \quad
   \inference{}{A \to B <: \top}
\end{gather*}

\begin{align*}
  A \approx B &\equiv\; A <: B \;\text{ and }\; B <: A
\end{align*}
\caption{Our subtyping relation on intersection types.}
\label{fig:subtyping}
\end{figure}

The set of types $\mathbb{T}$ is isomorphic to the set $\mathbb{D}$
defined in Section~\ref{sec:denot-function}. The singleton types are
isomorphic to integers. The function, intersection, and $\top$ types
taken together are isomorphic to function tables. Each entry in a
table corresponds to a function type. The following two functions,
$\mathsf{typof}$ and $\mathsf{eltof}$, witness the
isomorphism. Strictly speaking, the isomorphism is between
$\mathbb{D}$ and $\mathbb{T}/{\approx}$, so elements of $\mathbb{D}$
are a canonical form for types.
\begin{center}
\begin{minipage}[t]{0.45\textwidth}
\begin{align*}
\mathsf{typof} &: \mathbb{D} \to \mathbb{T} \\
\mathsf{typof}(n) &= n \\
\mathsf{typof}(t) &= \bigwedge_{(d,d')\in t} \mathsf{typof}(d){\to}\mathsf{typof}(d')
\end{align*}
\end{minipage}
\;
\begin{minipage}[t]{0.45\textwidth}
\begin{align*}
\mathsf{eltof} &: \mathbb{T} \to \mathbb{D} \\
\mathsf{eltof}(n) &= \{ n \} \\
\mathsf{eltof}(F) &= \mathsf{tabof}(F) \\[1ex]
\mathsf{tabof}(A \to B) &= \{ (\mathsf{eltof}(A), \mathsf{eltof}(B)) \} \\
\mathsf{tabof}(A \wedge B) &= \mathsf{eltof}(A) \cup \mathsf{eltof}(B) \\
\mathsf{tabof}(\top) &= \emptyset
\end{align*}
\end{minipage} 
\end{center}

The subtyping relation of Figure~\ref{fig:subtyping} is the inverse of
the $\sqsubseteq$ ordering on $\mathbb{D}$.

\begin{proposition}[Subtyping is inverse of $\sqsubseteq$ and is related to $\in$]\ 
\begin{enumerate}
\item If $A <: B$, then $\mathsf{eltof}(B) \sqsubseteq \mathsf{eltof}(A)$.
\item If $d \sqsubseteq d'$, then $\mathsf{typof}(d') <: \mathsf{typof}(d)$.
\item If $(d,d') \in t$, then $\mathsf{typof}(t) <: \mathsf{typof}(d) \to \mathsf{typof}(d')$.
\end{enumerate}
\end{proposition}

\begin{proposition}[Types and Elements are Isomorphic]\ \\
$\mathsf{eltof}(\mathsf{typof}(d)) = d$ and
$\mathsf{typof}(\mathsf{eltof}(A)) \approx A$
\end{proposition}

Figure~\ref{fig:denot-type-system} defines our elementary semantics as
a type system. The type system is the same as the one in
Figure~\ref{fig:intersection-type-system} except that it replaces
$\wedge$-introduction:
\[
  \inference{\Gamma \vdash e : A & \Gamma \vdash e : B}
            {\Gamma \vdash e : A \wedge B}
\]
with a more specialized rule that requires $e$ to be a $\lambda$
abstraction.


\begin{figure}[tbp]
\hfill\fbox{$\Gamma \vdash e : A$}
\begin{gather*}
\cdots \qquad
\inference{\Gamma \vdash \lam{x} e : A & \Gamma \vdash \lam{x} e : B}
          {\Gamma \vdash \lam{x} e : A \wedge B}
\end{gather*}
\caption{The new elementary semantics for CBV $\lambda$-calculus as
  a type system.}
\label{fig:denot-type-system}
\end{figure}



\begin{theorem}[Equivalence of type system and elementary semantics]\ 
\begin{enumerate}
\item $d \in \SEM{e}\rho$ implies
  $\mathsf{typof}(\rho) \vdash e : \mathsf{typof}(d)$ and
\item $\Gamma \vdash e : A$ implies
  $\mathsf{eltof}(A) \in \SEM{e}\mathsf{eltof}(\Gamma)$.
\end{enumerate}
where $\mathsf{typof}(\rho)(x) = \mathsf{typof}(\rho(x))$
and $\mathsf{eltof}(\Gamma)(x) = \mathsf{eltof}(\Gamma(x))$.
\end{theorem}

\section{Mechanized Correctness of Elementary Semantics}
\label{sec:correct}

We prove that our elementary semantics for CBV $\lambda$-calculus is
equivalent to the standard operational semantics on programs.
Section~\ref{sec:sound-wrt-op-sem} proves one direction of the
equivalence and Section~\ref{sec:complete-wrt-op-sem} proves the
other. Then to justify the use of the elementary semantics in compiler
optimizations, we prove in
Section~\ref{sec:sound-wrt-contextual-equiv} that it is sound with
respect to contextual equivalence.

\subsection{Sound with respect to operational semantics}
\label{sec:sound-wrt-op-sem}

In this section we prove that the elementary semantics is sound with
respect to the relational semantics of \citet{Kahn:1987aa}, written
$\rho \vdash e \Rightarrow v$. For each lemma and theorem we give a
proof sketch that includes which other lemmas or theorems were
needed. For the full details we refer the reader to the Isabelle
mechanization.

We relate $\mathbb{D}$ to sets of syntactic values $\mathbb{V}$ of
\citet{Kahn:1987aa} with the following logical relation. Logical
relations are usually type-indexed, not element-indexed, but our
domain elements are isomorphic to types.
\begin{align*}
\mathcal{G} & : \mathbb{D} \to \SET{\mathbb{V}} \\
\mathcal{G}(n) &= \{ n \} \\
\mathcal{G}(t) &= \left\{ \langle \lam{x}e,\varrho\rangle \middle|
   \begin{array}{l}
   \forall (d_1,d_2) {\in} t. \forall v_1.\, v_1{\in} \mathcal{G}(d_1) \implies \\
   \exists v_2.\, \varrho(x{:=}v_1) \vdash e \Rightarrow v_2 
      \land v_2{\in} \mathcal{G}(d_2) 
   \end{array}
   \right\} 
\end{align*}

\begin{lemma}[$\mathcal{G}$ is downward closed]
\label{lem:sub-good}
If $v \in \mathcal{G}(d)$ and $d' \sqsubseteq d$, then $v \in \mathcal{G}(d')$.
\end{lemma}
\begin{proof}
The proof is a straightforward induction on $d$.
\end{proof}

We relate the two environments with the inductively defined predicate
$\mathcal{G}(\rho,\varrho)$.
\begin{gather*}
  \inference{}{\mathcal{G}(\emptyset, \emptyset)}
  \quad
  \inference{v \in \mathcal{G}(d) & \mathcal{G}(\rho,\varrho)}
            {\mathcal{G}(\rho(x{:=}d), \varrho(x{:=}v))}
\end{gather*}

\begin{lemma}
\label{lem:lookup-good}
  If $\mathcal{G}(\rho,\varrho)$,
  then $\varrho(x) \in \mathcal{G}(\rho(x))$
\end{lemma}
\begin{proof}
  The proof is by induction on the derivation of
  $\mathcal{G}(\rho,\varrho)$.
\end{proof}

\begin{lemma}
  \label{lem:denot-terminates}
  If $d \in \SEM{e}\rho$ and $\mathcal{G}(\rho,\varrho)$,
  then $\varrho \vdash e \Rightarrow v$ and $v \in \mathcal{G}(d)$
  for some $v$.
\end{lemma}
\begin{proof}
The proof is by induction on $e$.  The cases for integers and
arithmetic operations are straightforward.  The case for variables
uses Lemmas~\ref{lem:sub-good} and \ref{lem:lookup-good}.  The cases
for $\lambda$ abstraction and application use
Lemma~\ref{lem:sub-good}.
\end{proof}

\begin{theorem}[Sound wrt. op. sem.]
  \label{thm:sound-wrt-op-sem} (aka. Adequacy)\\
  If $\SEM{e}\emptyset = \SEM{n}\emptyset$, then $e \Downarrow n$.
\end{theorem}
\begin{proof}
From the premise we have $n \in \SEM{e}\emptyset$.
Also we immediately have $\mathcal{G}(\emptyset,\emptyset)$.
So by Lemma~\ref{lem:denot-terminates} we have
$\emptyset \vdash e \Rightarrow v$ and $v \in \mathcal{G}(n)$
for some $v$. So $v = n$ and therefore $e \Downarrow n$.
\end{proof}

This proof of soundness can also be viewed as a proof of
implementation correctness.  The operational semantics, being
operational, can be viewed as a kind of implementation, and in this
light, the above proof is an example of how convenient it can be to
prove (one direction of) implementation correctness using the
elementary semantics as the specification.

\subsection{Complete with respect to operational semantics}
\label{sec:complete-wrt-op-sem}

In this section we prove that the elementary semantics is complete
with respect to the small-step semantics for the CBV
$\lambda$-calculus~\citep{Pierce:2002hj}. The proof strategy is
adapted from work by \citet{Alessi:2003aa} on intersection types.  We
need to show that if $e \longrightarrow^{*} n$, then $\SEM{e}\emptyset
= \SEM{n}\emptyset$. The meaning of the last expression in the
reduction sequence is equal to $\SEM{n}\emptyset$, so if we could just
walk this backwards, one step at a time, we would have our
result. That is, we need to show that if $e \longrightarrow e'$, then
$\SEM{e}\rho = \SEM{e'}\rho$.  We can decompose the equality into
$\SEM{e}\rho \subseteq \SEM{e'}\rho$ and $\SEM{e'}\rho \subseteq
\SEM{e}\rho$.  The forward direction is equivalent to proving type
preservation for our intersection type system, which is
straightforward. Let us focus on the backward direction and the case
of $\beta$ reduction.

Consider the following example reduction:
\[
  (\lam{x} (x \app 1) + (x \app 2)) \app (\lam{y} \ldots)
  \longrightarrow
  ((\lam{y} \ldots) \app 1) + ((\lam{y} \ldots) \app 2)
\]
where $e=(\lam{x} \ldots) \app (\lam{y} \ldots)$ and $e'=((\lam{y}
\ldots) \app 1) + ((\lam{y} \ldots) \app 2)$.  For some arbitrary $d$,
we can assume that $d \in \SEM{e'}\emptyset$ and need to show $d \in
\SEM{e}\emptyset$. From $d \in \SEM{e'}\emptyset$ we know there must
have been some tables $t_1$ and $t_2$ such that $t_1 \in
\SEM{\lam{y}\ldots}\emptyset$ and $t_2 \in
\SEM{\lam{y}\ldots}\emptyset$, but we only know for sure that $1$ is
in the domain of $t_1$ and $2$ is in the domain of $t_2$.  Perhaps
$t_1 = \{ (1,7)\}$ and $t_2 = \{ (2,0) \}$.  However, to obtain $d \in
\SEM{e}\emptyset$ we need a single table $t_3$ that can be bound to
variable $x$, and has both $1$ and $2$ in its domain so that the
applications $(x \app 1)$ and $(x \app 2)$ make sense. Fortunately, we
can simply combine the two tables (Lemma~\ref{lem:combine-values}).
\begin{equation}
  t_3=t_1 \cup t_2 = \{ (1,7), (2,0) \} \label{eq:t3}
\end{equation}

In general, we also need a lemma about reverse substitution.  Recall
the rule for $\beta$ reduction
\[
(\lam{x} e_1) \app v \longrightarrow e_1[v/x]
\]
Then we need a lemma that says, for any $d_1$,
there is some $d_2 \in \SEM{v}\emptyset$ such that
\[
d_1 \in \SEM{e_1[v/x]}\emptyset
\quad\text{implies}\quad
d_1 \in \SEM{e_1}\emptyset(x\by d_2).
\]
Generalizing this so that it can be proved by induction gives us
Lemma~\ref{lem:reverse-subst-pres-denot} below.

We proceed with the formal development of the completeness proof.
%
%
We define equivalence of environments, written $\rho \approx \rho'$,
and note that $\rho \approx \rho'$ implies $\rho \sqsubseteq \rho'$.
\[
\rho \approx \rho' \equiv \forall x.\, \rho(x) = \rho'(x) 
\]

\begin{lemma}[Reverse substitution preserves meaning]
\label{lem:reverse-subst-pres-denot}
  If $d \in \SEM{e[\val{v}/y]}\rho$,
  then $d \in \SEM{e}\rho'$,
  $d' \in \SEM{\val{v}}\emptyset$,
  and $\rho' \approx \rho(y{:=}d')$ for some $\rho',d'$.
\end{lemma}
\begin{proof}\ 
The proof is by induction on $e'=e[\val{v}/y]$.  But before
considering the cases for $e'$, we first consider whether or not
$e=y$. The proof uses Propositions \ref{prop:change-env-le} and
\ref{prop:e-val} and Lemma~\ref{lem:combine-values}.
\end{proof}


\begin{lemma}[Reverse reduction preserves meaning]\ 
\label{lem:reverse_step_pres_denot}
\label{lem:reverse-multi-step-pres-denot}
\begin{enumerate}
\item If $e \longrightarrow e'$, then $\SEM{e'}\rho \subseteq \SEM{e}\rho$.
\item If $e \longrightarrow^{*} e'$, then $\SEM{e'}\rho \subseteq
  \SEM{e}\rho$.
\end{enumerate}
\end{lemma}
\begin{proof}\ 
  \begin{enumerate}
  \item The proof is by induction on the derivation of $e
    \longrightarrow e'$.  All of the cases are straightforward except
    for $\beta$ reduction. In that case we have
    $(\lam{x}e_1)\app \val{v} \longrightarrow e_1[\val{v}/x]$.
    Fix an arbitrary $d$ and assume $d \in \SEM{e_1[\val{v}/x]}\rho$.
    We need to show that $d \in \SEM{(\lam{x}e_1)\app \val{v}}\rho$.
    By Lemma~\ref{lem:reverse-subst-pres-denot} and the assumption
    there exist $d'$ and $\rho'$ such that
    $d \in \SEM{e_1}\rho'$, 
    $d' \in \SEM{\val{v}}\emptyset$, and 
    $\rho' \approx \rho(x\by d')$. 
    Then we have $d \in \SEM{e_1}\rho(x\by d')$
    by Proposition~\ref{prop:change-env-le}, noting that
    $\rho' \sqsubseteq \rho(x\by d')$. Therefore we have
    $\{ (d',d)\} \in \SEM{\lam{x}e_1}\rho$. 
    Also, we have $d' \in \SEM{\val{v}}\rho$
    by another use of Proposition~\ref{prop:change-env-le},
    noting that $\emptyset \sqsubseteq \rho$.
    With these two facts, we conclude that
    $d \in \SEM{(\lam{x}e_1)\app \val{v}}\rho$.
    
  \item The proof is by induction on the derivation of $e
    \longrightarrow^{*} e'$. The base case is trivial and the
    induction step follows immediately from part (1).
  \end{enumerate}
\end{proof}

Next we prove the forward direction, that reduction preserves
meaning. The proof follows the usual pattern for preservation of a
type system. However, we shall continue to use the denotational
semantics here.

\begin{lemma}[Substitution preserves meaning]
 \label{lem:substitution}
  If $d \in \SEM{e}\rho'$, $d' \in \SEM{\val{v}}\emptyset$, and $\rho'
  \approx \rho(x\by d')$, then $d \in \SEM{e[\val{v}/x]}$
\end{lemma}
\begin{proof}
  The proof is by induction on $e$. The case for variables 
  uses Proposition~\ref{prop:change-env-le}.
  The case for $\lambda$ uses Proposition~\ref{prop:change-env-le}.
\end{proof}

\begin{lemma}[Reduction preserves meaning]\ 
  \label{lem:preservation}
  \begin{enumerate}
  \item If $e \longrightarrow e'$, then $\SEM{e}\rho \subseteq \SEM{e'}\rho$.
  \item If $e \longrightarrow^{*} e'$, then $\SEM{e}\rho\subseteq \SEM{e'}\rho$.
  \end{enumerate}
\end{lemma}
\begin{proof}\ 
\begin{enumerate}
\item The proof is by induction on $e \longrightarrow e'$.
  Most of the cases are straightforward.
  The case for $\beta$ uses Proposition~\ref{prop:change-env-le}
  and Lemma~\ref{lem:substitution}.
\item The proof is by induction on the derivation of $e \longrightarrow^{*} e'$,
  using part (1) in the induction step.
\end{enumerate}
\end{proof}

\begin{corollary}[Meaning is invariant under reduction]\ \\
  \label{lem:meaning-invarient-reduction}
  If $e \longrightarrow^{*} e'$, then $\SEM{e} = \SEM{e'}$.
\end{corollary}
\begin{proof}
  The two directions are proved by
  Lemma~\ref{lem:reverse-multi-step-pres-denot} and
  \ref{lem:preservation}.
\end{proof}

\begin{theorem}[Complete wrt. op. sem.]
  \label{thm:completeness}
  If $e \Downarrow n$, then $\SEM{e} \emptyset = \SEM{n}\emptyset$.
\end{theorem}
\begin{proof}
  From the premise we have $e \longrightarrow^{*} n$, from which we
  conclude by use of Lemma~\ref{lem:meaning-invarient-reduction}.
\end{proof}

The completeness theorem is rather important for the elementary
semantics. It says that $\CE$ gives the right meaning to all the
terminating programs in the CBV $\lambda$-calculus.

We also prove that $\CE$ gives the right meaning to diverging
programs. That is, $\CE$ maps diverging programs to $\emptyset$. We
write $e\Uparrow$ when $e$ diverges.

\begin{proposition}[Diverging programs have empty meaning]
\label{prop:diverge-denot-empty}\ \\
If $e\Uparrow$ then $\SEM{e}\emptyset = \emptyset$.
\end{proposition}
\begin{proof}
  Towards a contradiction, suppose $\SEM{e}\emptyset \neq \emptyset$.
  So $d \in \SEM{e}\emptyset$ for some $d$. We have
  $\mathcal{G}(\emptyset,\emptyset)$, so $\emptyset \vdash e
  \Rightarrow v$ by Lemma~\ref{lem:denot-terminates}.  Thus, we also
  have $e \longrightarrow^{*} v$.  Then
  from $e\Uparrow$, we have $v \longrightarrow e'$ for some
  $e'$, but that is impossible because $v$ is a value and so cannot
  further reduce.
\end{proof}

\noindent In contrast, syntactic values have non-empty meaning.

\begin{proposition}[Syntactic values have non-empty meaning]
  \label{prop:e-val}
  $\SEM{\val{v}}\rho \neq \emptyset$.
\end{proposition}
\begin{proof}
  The proof is by case analysis on $\val{v}$.
\end{proof}


\subsection{Sound with respect to contextual equivalence}
\label{sec:sound-wrt-contextual-equiv}

We would like to use our elementary semantics to justify compiler
optimizations, which replace sub-expressions within a program with
other sub-expressions (that are hopefully more efficient).  Two
sub-expressions are \emph{contextual equivalent}, defined below, when
replacing one with the other does not change the behavior of the
program~\citep{James-Hiram-Morris:1968kx}. We define contexts $C$ with
the following grammar.
\begin{align*}
  C & ::= \box \mid C \oplus e \mid e \oplus C \mid \lam{x}C \mid C \app e
     \mid e \app C & \text{contexts}\\
e \simeq e' &\equiv \forall C.\,
  FV(C[e]) = FV(C[e']) = \emptyset {\implies}
  C[e] \Downarrow \text{ iff } C[e'] \Downarrow & \text{ctx. equivalence}
\end{align*}

The correctness property that we are after is that denotational
equality should imply contextual equivalence.  A common way to prove
this is to show that the denotational semantics is a congruence and
then use soundness and completeness of the semantics for
programs~\citep{Gunter:1992aa}. Indeed, we take that approach.

\begin{lemma}[$E$ is a congruence]
  \label{lem:congruence}
  For any context $C$, 
  if $\SEM{e} = \SEM{e'}$, then $\SEM{C[e]} = \SEM{C[e']}$.
\end{lemma}
\begin{proof}
  The proof is a straightforward induction on $C$.
\end{proof}

\begin{theorem}[Sound wrt. Contextual Equivalence]\ \\
  \label{thm:denot-sound-wrt-ctx-equiv}
  If $\SEM{e} = \SEM{e'}$,
  then $C[e] \Downarrow$ iff $C[e'] \Downarrow$
  for any closing context $C$.
\end{theorem}
\begin{proof}
  We discuss one direction of the iff, that $C[e]\Downarrow$ implies
  $C[e']\Downarrow$. The other direction is similar.  From the premise,
  congruence gives us $\SEM{C[e]} = \SEM{C[e']}$.  From
  $C[e]\Downarrow$ we have $C[e] \longrightarrow^{*} \val{e}$ for some
  $\val{e}$.  Therefore we have $\SEM{C[e]}\emptyset =
  \SEM{\val{e}}\emptyset$ by completeness
  (Theorem~\ref{thm:completeness}).  Then we also have
  $\SEM{C[e']}\emptyset = \SEM{\val{e}}$.  So by
  Proposition~\ref{prop:e-val}, we have $v \in \SEM{\val{e}}$ for
  some $v$, and therefore $v \in \SEM{C[e']}$.  We conclude that
  $\emptyset \vdash C[e'] \Rightarrow w$ for some $w$ by soundness
  (Lemma~\ref{lem:denot-terminates}).
\end{proof}

\section{Mechanized Correctness of an Optimizer}
\label{sec:optimize}

We turn to address the question of whether the declarative semantics
is useful. Our first case study is proving the correctness of a
compiler optimization pass: constant folding and function
inlining. The setting is still the untyped $\lambda$-calculus extended
with integers and arithmetic. Figure~\ref{fig:optimizer} defines the
optimizer as a function $\mathcal{O}$ that maps an expression and a
counter to an expression. One of the challenging problems in creating
a good inliner is determining when to stop. Here we use a counter that
limits inlining to a fixed depth $k$. A real compiler would use a
smarter heuristic but it would employ the same program
transformations.

The third equation in the definition of $\mathcal{O}$ performs
constant folding. For an arithmetic operation $e_1 \oplus e_2$, it
recursively optimizes $e_1$ and $e_2$. If the results are integers,
then it performs the arithmetic. Otherwise it outputs an arithmetic
expression.
The fourth equation optimizes the body of a $\lambda$ abstraction.
The fifth equation, for function application, is the most interesting.
If $e_1$ optimizes to a $\lambda$ abstraction and $e_2$ optimizes to a
syntactic value $v_2$, then we perform inlining by substituting $v_2$
for parameter $x$ in the body of the function. We then optimize the
result of the substitution, making this a rather aggressive
polyvariant
optimizer~\citep{Jagannathan:1996aa,Waddell:1997fk,Banerjee:1997aa,Turbak:1997aa}.
The counter is decremented on this recursive call to ensure
termination.

\begin{figure}[tbp]
\begin{align*}
  \mathcal{O}\EXP{ x }k &= x \\
  \mathcal{O}\EXP{ n }k &= n \\
  \mathcal{O}\EXP{ e_1 \oplus e_2 }k &=
  \begin{cases}
    n_1 \oplus n_2 & \text{if }
    \mathcal{O}\EXP{e_1 }k = n_1 \text{ and }
    \mathcal{O}\EXP{e_2 }k = n_2 \\
    \EXP{\mathcal{O}\EXP{e_1 }k \oplus \mathcal{O}\EXP{e_2 }k}
    & \text{otherwise}
  \end{cases} \\
  \mathcal{O}\EXP{ \lambda x.\, e }k &=
    \lambda x.\, \mathcal{O}\EXP{ e }k \\
  \mathcal{O}\EXP{ e_1\,e_2 }k &=
  \begin{cases}
    \mathcal{O}\EXP{ e[v_2/x] }(k{-}1)
    & \text{if } k \geq 1 \text{ and } \mathcal{O}\EXP{e_1 }k = \lambda x.\, e \\
     & \text{and } \mathcal{O}\EXP{e_2 }k = v_2 \\
    \EXP{\mathcal{O}\EXP{e_1 }k \app \mathcal{O}\EXP{e_2 }k} & \text{otherwise}
  \end{cases} 
\end{align*}

\caption{A compiler optimization pass that folds and propagates
  constants and inlines function calls.}
\label{fig:optimizer}
\end{figure}

We turn to proving the optimizer correct with respect to the
declarative semantics. The proof is pleasantly straightforward!

\begin{lemma}[Optimizer Preserves Denotations]
  \label{lem:O2-correct-aux}
  $\CE(\mathcal{O}\EXP{e}k) = \SEM{e}$
\end{lemma}
\begin{proof}
  The proof is by induction on the termination metric for
  \(\mathcal{O}\), which is the lexicographic ordering of \(k\) then
  the size of \(e\). All the cases are straightforward to prove
  because reduction preserves meaning
  (Lemma~\ref{lem:meaning-invarient-reduction}) and because meaning is
  a congruence (Lemma~\ref{lem:congruence}). 
\end{proof}

\noindent The mechanized proof of Lemma~\ref{lem:O2-correct-aux} is
under 30 lines!

\begin{theorem}[Correctness of the Optimizer]
  \label{thm:O2-correct}
  $e \simeq \mathcal{O}\EXP{e}k$
\end{theorem}
\begin{proof}
  The proof follows immediately from the above
  Lemma~\ref{lem:O2-correct-aux} and soundness with respect to
  contextual equivalence
  (Theorem~\ref{thm:denot-sound-wrt-ctx-equiv}).
\end{proof}

\section{Elementary Semantics in a Non-determinism Monad}
\label{sec:non-determinism-monad}


Toward making the elementary semantics easier to scale to larger
languages, we show how to hide the set-valued aspect of the denotation
function $\SEM{e}\rho$ behind a non-determinism
monad~\citep{wadler92:_essence}.  Recall that a non-deterministic
computation returns not just one result, but a set of results.
%
Non-deterministic choice is provided by the $\mathsf{bind}$ operation.
It chooses, one at a time, an element $a$ from a sub-computation $m$,
and proceeds with another sub-computation $f$ that depends on $a$,
then collect all the results into a set. The following is the
definition of $\mathsf{bind}$ together with short-hand notation.
\begin{align*}
  \mathsf{bind} & : \SET{\alpha} \to
        (\alpha \to \SET{\beta}) \to \SET{\beta} \\[-0.5ex]
  \mathsf{bind}\;m\;f &=
      \{ b \mid \exists a. \, a \in m \land b \in f(a) \} \\[-0.5ex]
  X \leftarrow m_1; m_2 &\equiv
         \mathsf{bind}\;m_1\; (\lambda X.\, m_2)
\end{align*}
Backtracking is provided by the way the $\mathsf{zero}$ operation
interacts with $\mathsf{bind}$. The $\mathsf{zero}$ operation simply
says to ``fail'' or ``abort'' by returning an empty set.
\[
  \mathsf{zero} : \SET{\alpha} \qquad
  \mathsf{zero} = \emptyset
\]
As usual, the monad also provides a $\mathsf{return}$ operation to
inject a result into the monad, in this case producing a (singleton)
set.
\[
  \mathsf{return} : \alpha \to \SET{\alpha} \qquad
  \mathsf{return}\,a = \{ a \} 
\]

With these monad operations in hand, we create some auxiliary
functions that are needed for the interpreter.  The following
$\mathsf{mapM}$ function applies a function $f$ to each element of a
finite set, producing a finite set, all within the non-determinism
monad. (We write $\uplus$ for the union of two sets that have no
elements in common.)
\begin{align*}
  \mathsf{mapM} &: \FSET{\alpha} \to(\alpha\to\SET{\beta}) \to
        \SET{\FSET{\beta}}\\[-0.5ex]
  \mathsf{mapM}\;\emptyset\;f  &= \return \emptyset \\[-0.5ex]
  \mathsf{mapM}\;(\{ a \} \uplus as)\;f  &=
    b \leftarrow f(a); bs \leftarrow \mathsf{mapM}\;as\;f; \return \{b\}\cup bs
\end{align*}
The next auxiliary function makes sure that the semantics is downward
closed with respect to $\sqsubseteq$. The function $\mathsf{down}\,d$
chooses an arbitrary element $d'$ and returns it if $d' \sqsubseteq d$
and otherwise backtracks to pick another element.
\begin{align*}
  \down &: \mathbb{D} \to \SET{\mathbb{D}} \\
  \down d &= d' \leftarrow \mathbb{D}; 
       \IF d' \sqsubseteq d \THEN \return d' \ELSE \zero
\end{align*}

The non-deterministic interpreter for the CBV $\lambda$-calculus is
defined in Figure~\ref{fig:denot-interp}. Let us focus on the cases
for $\lambda$ abstractions and function application. The cases for
integers, arithmetic, and conditionals look just as they would for any
interpreter in monadic style.
For $\SEM{\lam{x}e}\rho$, we non-deterministically choose a finite
domain $ds$ and then map over it to produce the function's table.  For
each input $d$, we interpret the body $e$ in an environment extended
with $x$ bound to $d$. This produces the output $d'$, which we pair
with $d$ to form one entry in the table.
For application $\SEM{e_1\app e_2}\rho$, we interpret $e_1$ and $e_2$
to the elements $d_1$ and $d_2$, then check whether $d_1$ is a
function table. If it is, we non-deterministically select an entry
$(d,d')$ from the table and see if the argument $d_2$ matches the
parameter $d$.  If so we return $\down d'$, otherwise we backtrack and
try another entry in the table or possibly backtrack even further and
try another table altogether!

\begin{figure}[tbp]
\hfill\fbox{$\SEM{e}\rho$}
\begin{align*}
\SEM{ n }\rho &= \mathsf{return}\,n \\
\SEM{e_1 \oplus e_2}\rho &= 
  \begin{array}[t]{l}
  d_1 \leftarrow \SEM{e_1}\rho;\;
  d_2 \leftarrow \SEM{e_2}\rho;\\
  \begin{array}[t]{l}
  \CASE (d_1,d_2) \OF
  \;  (n_1,n_2) \Rightarrow \mathsf{return}\, n_1 \oplus n_2 \;
  |\; \_ \Rightarrow \zero
  \end{array}
  \end{array} \\
\SEM{ x }\rho &= \down \rho(x) \\
\SEM{ \lam{x} e }\rho &=
    ds \leftarrow \FSET{\mathbb{D}};
    \mathsf{mapM}\;ds\;(\lam{d}
      d' \leftarrow \SEM{e}\rho(x\by d); \mathsf{return}\,(d,d')) \\
\SEM{ e_1\; e_2 } \rho &=
  \begin{array}[t]{l}
  d_1 \leftarrow \SEM{e_1}\rho;\;
  d_2 \leftarrow \SEM{e_2}\rho;\\
  \CASE d_1 \OF \\
   \;\, t \Rightarrow (d,d') \leftarrow t; \;
   \IF d \sqsubseteq d_2 \THEN \down d' \ELSE \zero  \\
  |\; \_ \Rightarrow \zero
  \end{array} \\
  \SEM{ \IF e_1 \THEN e_2 \ELSE e_3}\rho &=
\begin{array}[t]{l}
d_1 \leftarrow \SEM{e_1}\rho;\\
\CASE d_1 \OF\\
 \;\,n \Rightarrow
   \IF n \neq 0 \THEN \SEM{e_2}\rho \ELSE \SEM{e_3}\\
\mid \_ \Rightarrow \zero
\end{array}
\end{align*}
\caption{Elementary semantics for CBV $\lambda$-calculus as
  a non-deterministic interpreter.}
\label{fig:denot-interp}
\end{figure}


\section{Elementary Semantics and Soundness for System F}
\label{sec:cannot-go-wrong}

This section serves two purposes: it demonstrates that our elementary
semantics is straightforward to extend to a typed language with
first-class polymorphism, and it demonstrates how to use the
elementary semantics to prove type soundness. Our setting is the
polymorphic $\lambda$-calculus (System F)~\citep{GIRARD72,REYNOLDS74C}
extended with general recursion ($\mathtt{fix}$). The idea is that
this section is a generalization of Milner's ``well-typed expressions
do not go wrong''~\citep{Milner:1978kh}.

\subsection{Static Semantics}

The syntax and type system of this language is defined in
Figure~\ref{fig:system-f}. We represent type variables using DeBruijn
indices, with $\forall$ and $\Lambda$ acting as implicit binding
forms. The shift operation $\shift{k}{c}{A}$ increases the DeBruijn
indices greater or equal to $c$ in type $A$ by $k$ and the
substitution operation $[i\mapsto B]A$ replaces DeBruijn index $i$
within type $A$ with type $B$~\citep{Pierce:2002hj}.

\begin{figure}[tbp]
\[
\begin{array}{rcll}
  i,j  & \in & \mathbb{N} & \text{DeBruijn indices}\\
  A,B,C & ::= & \INT \mid A \to B \mid \forall\,A \mid i
   \qquad & \text{types} \\
 e & ::= & n \mid e \oplus e \mid x \mid \lam{x\of A} e \mid e\; e \mid \Lambda e
 \mid e [A] \mid \fix{x\of A} e \qquad & \text{expressions} 
\end{array}
\]
\begin{gather*}
  \Gamma \vdash n : \INT 
  \qquad
  \frac{}
       {\Gamma \vdash x : \mathsf{lookup}(\Gamma,x)}
   \\[2ex]
   \frac{\ext{\Gamma}{x}{A} \vdash e : B}
        {\Gamma \vdash \lam{x\of A} e : A \to B}
   \quad
   \frac{\ext{\Gamma}{x}{A\to B} \vdash e : A \to B }
        {\Gamma \vdash \fix{x\of A{\to}B} e : A \to B}  
   \quad
   \frac{\Gamma \vdash e : A \to B \quad
         \Gamma \vdash e' : A}
        {\Gamma \vdash e \; e' : B}
   \\[2ex]
   \frac{\tyext{\Gamma} \vdash e : A}
        {\Gamma \vdash \Lambda e : \forall A}
   \qquad
   \frac{\Gamma \vdash e : \forall A}
        {\Gamma \vdash e[B] : [0\mapsto B]A}
\end{gather*}

\caption{Syntax and type system of System F extended with general
  recursion (via \texttt{fix}).}
\label{fig:system-f}
\end{figure}

Type environments and their operations deserve some explanation in the
way they handle type variables. A \emph{type environment} $\Gamma$ is
a pair consisting of 1) a mapping for term variables and 2) a natural
number representing the number of enclosing type-variable binders. The
mapping is from variables names to a pair of a) the variable's type
and b) the number of enclosing type-variable binders at the variable's
point of definition. We define the operations $\mathsf{extend}$,
$\mathsf{tyExtend}$, and $\mathsf{lookup}$ in the following way.
\begin{align*}
  \ext{\Gamma}{x}{A} &\equiv (\fst{\Gamma}(x\by (A,\snd{\Gamma}), \; \snd{\Gamma}) \\
  \tyext{\Gamma} & \equiv (\fst{\Gamma}, \; \snd{\Gamma} + 1) \\
  \lookup{\Gamma}{x} & \equiv \shift{k-j}{0}{A}
  \qquad \text{where}\; (A,j) = \fst{\Gamma}(x) 
     \text{ and } k = \snd{\Gamma}
\end{align*}
Thus, the $\mathsf{lookup}$ operation properly transports the type of
a variable from its point of definition to its occurrence by
shifting its type variables the appropriate amount.  We say that a
type environment $\Gamma$ is \emph{well-formed} if
$\snd{\fst{\Gamma}(x)} \leq \snd{\Gamma}$ for all $x \in
\text{dom}(\fst{\Gamma})$.

\subsection{Denotational Semantics}

The domain differs from that of the $\lambda$-calculus in two
respects. First, we add a $\mathsf{wrong}$ element to represent a
runtime type error, just like \citet{Milner:1978kh}.  Second, we add
an element to represent type abstraction. From a runtime point of
view, a type abstraction $\Lambda e$ simply produces a thunk. That is,
it delays the execution of expression $e$ until the point of type
application (aka. instantiation). A thunk contains an optional
element: if the expression $e$ does not terminate, then the thunk
contains $\mathsf{none}$, otherwise the thunk contains the result of
$e$.
\[
\begin{array}{rcll}
  o & ::= & \mathsf{none} \mid \mathsf{some}(d) & \text{optional elements} \\
  t & ::= & \{ (d_1,d'_1),\ldots,(d_n,d'_n) \} & \text{tables}\\
  d & ::= & n \mid t \mid \abs{o} \mid \mathsf{wrong}
  \qquad & \text{elements}
\end{array}
\]

Building on the non-determinism monad
(Section~\ref{sec:non-determinism-monad}), we add support for the
short-circuiting due to errors with the following alternative form of
bind.
\[
  X := E_1; E_2 \equiv
  X \leftarrow E_1;
    \mathbf{if} \, X = \mathsf{wrong} \,\mathbf{then}\, \mathsf{return}\,\mathsf{wrong}
    \;\mathbf{else}\; E_2
\]

We define an auxiliary function $\mathsf{apply}$ to give the semantics
of function application.  The main difference with respect to
Figure~\ref{fig:denot-interp} is returning $\wrong$ when $e_1$
produces a result that is not a function table. We also replace uses of
$\leftarrow$ with $:=$ to short-circuit the computation in case
$\wrong$ is produced by a sub-computation.
\begin{align*}
  \mathsf{apply}(D_1,D_2) &= 
  \begin{array}[t]{l}
  d_1 := D_1;\;
  d_2 := D_2;\\
  \CASE d_1 \OF
   \; t \Rightarrow (d,d') \leftarrow t; \;
   \IF d \sqsubseteq d_2 \THEN \mathsf{down}\,d' \ELSE \zero  \\
  |\; \_ \Rightarrow \return \mathsf{wrong}  
  \end{array} 
\end{align*}

An elementary semantics for System F extended with \texttt{fix} is
defined in Figure~\ref{fig:denot-system-f}. Regarding integers and
arithmetic, the only difference is that arithmetic operations return
$\wrong$ when an input in not an integer. Regarding $\lambda$
abstraction, the semantics is the same as for the untyped
$\lambda$-calculus (Figure~\ref{fig:denot-interp}).

To give meaning to $(\fix{x\of A}e)$ we form its ascending Kleene
chain but never take its supremum. In other words, we define a
auxiliary function $\mathsf{iterate}$ that starts with no tables
($\emptyset$) and then iteratively feeds the function to itself some
finite number of times, produces ever-larger sets of tables that
better approximate the function. The parameter $L$ is the meaning
function $\mathcal{E}$ partially applied to the expression $e$. We
fully apply $L$ inside of $\mathsf{iterate}$, binding $x$ to the
previous approximations.

The declarative semantics of polymorphism is straightforward.  The
meaning of a type abstraction $\Lambda e$ is $\abs{\mathsf{some}(d)}$
if $e$ evaluates to $d$. The meaning of $\Lambda e$ is
$\abs{\mathsf{none}}$ if $e$ diverges. The meaning of a type
application $e[A]$ is to force the thunk, that is, it is any element
below $d$ if $e$ evaluates to $\abs{\mathsf{some}(d)}$. On the other
hand, if $e$ evaluates to $\abs{\mathsf{none}}$, then the type
application diverges. Finally, if $e$ evaluates to something other
than a thunk, the result is $\wrong$.

\begin{figure}[tbp]
\begin{align*}
\SEM{ n }\rho &= \mathsf{return}\,n \\
\SEM{e_1 \oplus e_2}\rho &= 
  \begin{array}[t]{l}
  d_1 := \SEM{e_1}\rho;\;
  d_2 := \SEM{e_2}\rho;\\
  \CASE (d_1,d_2) \OF
    (n_1,n_2) \Rightarrow \mathsf{return}\, n_1 \oplus n_2 \;
  |\; \_ \Rightarrow \mathsf{return}\, \mathsf{wrong}  
  \end{array} \\
\SEM{ x }\rho &= \mathsf{down}\,\rho(x) \\
\SEM{ \lam{x\of A} e }\rho &=
    ds \leftarrow \FSET{\mathbb{D}};
    \mathsf{mapM}\;ds\;(\lam{d}
      d' := \SEM{e}\rho(x\by d); \mathsf{return}\,(d,d')) \\   
\SEM{ e_1\; e_2 } \rho &=
   \mathsf{apply}(\SEM{ e_1 }\rho, 
          \SEM{ e_2 }\rho) \\
\SEM{ \fix{x\of A}e } \rho &=
        k \leftarrow \mathbb{N}; \mathsf{iterate}(k, x, \SEM{ e }, \rho) \\
  & \quad \textbf{where}
      \begin{array}[t]{l}
        \mathsf{iterate}(0, x, L, \rho) = \zero \\
        \mathsf{iterate}(k+1, x, L, \rho) = 
            d := \mathsf{iterate}(k, x, L, \rho);
            L \, \rho(x\by d)
      \end{array}\\
\SEM{ \Lambda e } \rho &=
\begin{array}[t]{l}
 \IF \SEM{ e }\rho = \emptyset \THEN \return \abs{\mathsf{none}} \\
 \ELSE d := \SEM{ e } \rho; \return \abs{\mathsf{some}(d)}
\end{array} \\
\SEM{ e[A] } \rho &=
    x := \SEM{ e } \rho;
  \begin{array}[t]{l}
   \CASE x \OF \\
    \;\; \abs{\mathsf{none}} \Rightarrow \zero \\
    \mid \abs{\mathsf{some}(d)} \Rightarrow \mathsf{down}\,d \\
    \mid \_ \Rightarrow \return \mathsf{wrong}
  \end{array}
\end{align*}

\caption{An elementary semantics for System F with general recursion.}
  \label{fig:denot-system-f}
\end{figure}

One strength of this elementary semantics is that it enables the use
of monads to implicitly handle error propagation, which is important
for scaling up to large language
specifications~\citep{Chargueraud:2013aa,Owens:2016aa,Amin:2017aa}.

\subsection{Semantics of Types}

We define types in the domain $\SET{\mathbb{D}}$ with the meaning
function $\mathcal{T}$ defined in Figure~\ref{fig:system-f-types}.
The meaning of $\INT$ is the integers.  To handle type variables,
$\mathcal{T}$ has a second parameter $\eta$ that maps each DeBruijn
index to a set of elements, i.e., the meaning of a type variable. So
the meaning of index $i$ is basically $\eta_i$.  The
$\mathsf{cleanup}$ function serves to make sure that $\mathcal{T}$ is
downward closed and that it does not include $\wrong$. (Putting the
use of $\mathsf{cleanup}$ in $\TSEM{i}$ instead of $\TSEM{\forall A}$
makes for a slightly simpler definition of well-typed environments,
which is defined below.)  We write $|\eta|$ for the length of a
sequence $\eta$.  The meaning of a function type $A \to B$ is all the
finite tables $t$ in which those entries with input in $\TSEM{A}$ have
output in $\TSEM{B}$.

Last but not least, the meaning of a universal type $\forall A$
includes $\abs{\mathsf{none}}$ but also $\abs{\mathsf{some}(d)}$
whenever $d$ is in the meaning of $A$ but with $\eta$ extended with an
arbitrary set of elements.

\begin{figure}[tbp]
\begin{align*}
\TSEM{ \INT } \eta &= \mathbb{Z} \\
\TSEM{ i } \eta &= 
  \begin{cases}
      \mathsf{cleanup}(\eta_i) &\text{if } i < |\eta| \\
      \emptyset & \text{otherwise}
  \end{cases} \\
\TSEM{ A\to B } \eta &= 
  \{ t\mid 
    \forall (d_1,d'_2) \in t.\,
        d_1 \in \TSEM{A}\eta 
        \implies 
        \exists d_2.\, d_2 \in \TSEM{B}\eta \land
            d'_2 \sqsubseteq d_2
       \} \\
\TSEM{ \forall A } \eta &= 
  \{ \abs{\mathsf{some}(d)} \mid 
         \forall D.\, d \in \TSEM{A} D\eta \}
  \cup \{ \abs{\mathsf{none}} \} 
\end{align*}

\[
\mathsf{cleanup}(D) \equiv
 \{ d \mid \exists d'.\, d' \in D \land d \sqsubseteq d' \land d \neq \mathsf{wrong} \}
\]

\caption{A semantics of types for System F.}
\label{fig:system-f-types}
\end{figure}

We write $\vdash \rho,\eta : \Gamma$ to say that environments $\rho$
and $\eta$ are well-typed according to $\Gamma$ and define it
inductively as follows.
\begin{gather*}
 \vdash \emptyset,\emptyset : \emptyenv 
 \quad
  \inference{\vdash \rho,\eta : \Gamma \quad v \in \TSEM{ A } \eta}
       {\vdash \rho(x\by v),\eta : \ext{\Gamma}{x}{A} }
 \quad
  \inference{\vdash \rho,\eta : \Gamma}
       {\vdash \rho,V\eta : \tyext{\Gamma} }
\end{gather*}

\subsection{Type Soundness}
\label{sec:system-f-type-soundness}

We prove type soundness for System F, that is, if a program is
well-typed and has type $A$, then its meaning (in a well-typed
environment) is a subset of the meaning of its type $A$.

\begin{restatable}[Semantic Soundness]{thm}{cannotgowrong}\ \\
\label{thm:welltyped-dont-go-wrong}
If $~\Gamma \vdash e : A$ and $\vdash \rho,\eta : \Gamma$,
then
$\SEM{e} \rho \subseteq \TSEM{ A } \eta$.
\end{restatable}

\noindent We use four small lemmas. The first corresponds to
Milner's Proposition 1~\citep{Milner:1978kh}.

\begin{lemma}[$\mathcal{T}$ is downward closed]\ \\
\label{lem:T-down-closed}
If $d \in \TSEM{ A } \eta$ and $d' \sqsubseteq d$,
then $d' \in \TSEM{ A } \eta$.
\end{lemma}
\begin{proof}
  A straightforward induction on $A$
\end{proof}

\begin{lemma}[$\mathsf{wrong}$ not in $\mathcal{T}$]
\label{lem:wrong-not-in-T}
For any $A$ and $\eta$,
$\mathsf{wrong} \notin \TSEM{ A } \eta$.
\end{lemma}
\begin{proof}
A straightforward induction on $A$.
\end{proof}

\begin{lemma}
\label{lem:wfenv-good-ctx}
If $\vdash \rho : \Gamma$, 
then $\Gamma$ is a well-formed type environment.
\end{lemma}
\begin{proof}
This is proved by induction on the derivation of $\vdash \rho :
\Gamma$.
\end{proof}

\begin{lemma}
  \label{lem:lift-append-preserves-T}
  \label{lem:shift-cons-preserves-T}
  $
  \TSEM{ A } (\eta_1 \eta_3) =
  \TSEM{ \shift{|\eta_2|}{ |\eta_1|}{A} }
   (\eta_1\eta_2\eta_3)
  $
\end{lemma}
\begin{proof}
This lemma is proved by induction on $A$. 
\end{proof}

Next we prove one or two lemmas for each language feature.
Regarding term variables, we show that variable lookup is sound in a
well-typed environment.

\begin{lemma}[Variable Lookup]\ \\
\label{lem:lookup-wfenv}
If \ $\vdash \rho,\eta : \Gamma$ and $x \in \text{dom}(\Gamma)$,
then $\rho(x) \in \TSEM{ \mathsf{lookup}(\Gamma,x) } \eta$
\end{lemma}
\begin{proof}
The proof is by induction on the derivation of $\vdash \rho,\eta :
\Gamma$. The first two cases are straightforward but the third case
requires some work and uses lemmas~\ref{lem:wfenv-good-ctx} and
\ref{lem:shift-cons-preserves-T}.
\end{proof}

\noindent The following lemma proves that function application is
sound, similar to Milner's Proposition 2.

\begin{lemma}[Application cannot go wrong]\ \\
\label{lem:fun-app}
If $D \subseteq \TSEM{ A {\to} B } \eta$
and $D' \subseteq \TSEM{ A } \eta$, 
then $\mathsf{apply}(D,D') \subseteq \TSEM{ B } \eta$.
\end{lemma}
\begin{proof}
The proof is direct and uses lemmas~\ref{lem:wrong-not-in-T} and
\ref{lem:T-down-closed}.
\end{proof}

\noindent When it comes to polymorphism, our proof necessarily differs
considerably from Milner's, as we must deal with first-class
polymorphism. So instead of Milner's Proposition 4 (about type
substitution), we instead have a Compositionality Lemma analogous to
what you would find in a proof of
Parametricity~\citep{Skorstengaard:2015aa}. We need this lemma because
the typing rule for type application is expressed in terms of type
substitution ($e[B]$ has type $[0\mapsto B]A$) but the meaning of
$\forall A$ is expressed in terms of extending the environment $\eta$.

\begin{lemma}[Compositionality]
\label{lem:compositionality}
\[
  \TSEM{ A } (\eta_1 D \eta_2)
  = \TSEM{ [ |\eta_1| \mapsto B] A }  (\eta_1 \eta_2)
 \qquad \text{where } D = \TSEM{ B } (\eta_1\eta_2).
\]
\end{lemma}
\begin{proof}
The proof is by induction on $A$.  All of the cases were
straightforward except for $A=\forall A'$.  In that case we use the
induction hypothesis to show the following and then apply
Lemma~\ref{lem:shift-cons-preserves-T}.
\[
  \TSEM{ A' } (D \eta_1 D_B \eta_2)
  {=}
 \TSEM{ ([|D\eta_1|\mapsto \shift{1}{0}{B}} A' ] (D\eta_1\eta_2)
\text{ where }
D_B {=} \TSEM{ \shift{1}{0}{B} } (D\eta_1\eta_2)
\]
\end{proof}

\noindent The last lemma proves that $\mathsf{iterate}$ produces
tables of the appropriate type.

\begin{lemma}[Iterate cannot go wrong]
\label{lem:iterate-sound}
If
\begin{itemize}
\item  $d \in \mathsf{iterate}(k, x, L,\rho)$ and
\item for any $d'$, $d' \in \TSEM{ A \to B } \eta$
 implies $L\app \rho(x\by d')\app\eta \subseteq \TSEM{ A{\to} B } \eta$,
\end{itemize}
then $d \in \TSEM{ A{\to} B } \eta$.
\end{lemma}
\begin{proof}
This is straightforward to prove by induction on $k$.
\end{proof}

\noindent We proceed to the main theorem, that well-typed programs
cannot go wrong.

\cannotgowrong*
\begin{proof}
  The proof is by induction on the derivation of $\Gamma \vdash e :
  A$.  The case for variables uses Lemmas~\ref{lem:T-down-closed} and
  \ref{lem:lookup-wfenv}.  The case for application uses
  Lemma~\ref{lem:fun-app}.  The case for \texttt{fix} applies
  Lemma~\ref{lem:iterate-sound}, using the induction hypothesis to
  establish the second premise of that lemma. The case for type
  application uses Lemma~\ref{lem:compositionality}.

\end{proof}

\subsection{Comparison to Syntactic Type Safety}

The predominant approach to proving type safety of a programming
language is via \emph{progress and
  preservation}~\cite{Martin-Lof:1985aa,wright94:_type_soundness,Pierce:2002hj,Harper:2012aa}
over a small-step operational semantics. Recall that the Progress
Lemma says that a well-typed program can either reduce or it is a
value, but it is never stuck, which would correspond to a runtime type
error.
The strength of this syntactic approach is that the semantics does not
need to explicitly talk about runtime type errors, but nevertheless
they are ruled out by the Progress Lemma.

In comparison, the semantic soundness approach that we used in this
section relies on using an explicit $\wrong$ element to distinguish
between programs that diverge versus programs that encounter a runtime
type error. The downside of this approach is that the author of the
semantics could mix up $\wrong$ and $\zero$ (divergence) and the
Semantic Soundness (Theorem~\ref{thm:welltyped-dont-go-wrong}) would
still hold. However, a simple auditing of the semantics can catch this
kind of mistake. Also, we plan to investigate whether techniques such
as step-indexing~\citep{Siek:2012ac,Owens:2016aa,Amin:2017aa} could be
used to distinguish divergence from $\wrong$.


\subsection{From Type Soundness to Parametricity}

An exciting direction for future research is to use the elementary
semantics for proving Parametricity and using it to construct Free
Theorems, replacing the frame models in the work of
\citet{Wadler:1989fk}. The idea would be to to adapt $\mathcal{T}$,
our unary relation on $\mathbb{D}$, into $\mathcal{V}$, a binary
relation on $\mathbb{D}$. We would define the following logical
relation $\mathcal{R}$ in terms of $\mathcal{V}$ and the 
semantics $\mathcal{E}$.
\[
  \mathcal{R}\llbracket A \rrbracket \eta =
       \{ (e_1,e_2) \mid 
         \exists v_1 v_2.\; v_1 \in \SEM{e_1}\emptyset
         \land v_2 \in \SEM{e_1}\emptyset \land
           (v_1,v_2) \in \mathcal{V}\llbracket A \rrbracket\eta \}
\]
Then the Parametricity Theorem could be formulated as: 
\begin{quote}
If $\Gamma \vdash e : A$ and $\vdash \rho,\eta : \Gamma$, then $(e,e)
\in \mathcal{R}\llbracket A \rrbracket \eta$. 
\end{quote}
The proof would likely be similar to our proof of Semantic Soundness.

\section{Related Work}
\label{sec:related-work}

\paragraph{Intersection Type Systems}

The type system view of our elementary semantics
(Section~\ref{sec:denot-type-system}) is a variant of the intersection
type system invented by \citet{Coppo:1979aa} to study the untyped
$\lambda$-calculus.
Researchers have studied numerous properties and variations of the
intersection type
system~\citep{Hindley:1982aa,Coppo:1984aa,Coppo:1987aa,Alessi:2003aa,Alessi:2006aa}. Our
mechanized proof of completeness with respect to operational semantics
(Section~\ref{sec:complete-wrt-op-sem}) is based on a (non-mechanized)
proof by \citet{Alessi:2003aa}.

By making subtle changes to the subtyping relation it is possible to
capture alternate semantics~\citep{Alessi:2004aa} such as
call-by-value~\citep{Egidi:1992aa,Rocca:2004aa} or lazy
evaluation~\citep{Abramsky:1987aa}. \citet{Barendregt:2013aa} give a
thorough survey of these type systems. Our $\top$ type corresponds to
the type $\nu$ of \citet{Egidi:1992aa} and \citet{Alessi:2003aa}.  Our
subtyping relation is rather minimal, omitting the usual rules for
function subtyping and the distributive rule for intersections and
function types. Our study of singleton integer types within an
intersection type system appears to be a novel combination.

Intersection type systems have played a role in the full abstraction
problem for the lazy $\lambda$-calculus, in the guise of domain
logics~\citep{Abramsky:1987aa,Abramsky:1990vv,Jeffrey:1993aa}

The problem of inhabitation for intersection type systems has seen
recent progress~\citep{Dudenhefner:2017aa} and applications to
example-directed synthesis~\citep{Frankle:2016aa}.




\paragraph{Other Semantics}

There are many other approaches to programming language semantics that
we have not discussed, from axiomatic
semantics~\citep{Floyd:1969aa,Hoare:1969kw} to
games~\citep{Abramsky:2000kx,Hyland:2000ve}, event
structures~\citep{Winskel:1987fp}, and
traces~\citep{Reynolds:2004fk,Jeffrey:2005kx}.
Our function tables can be viewed as finitary versions of the tree
models for SPCF~\citep{Cartwright:1992ys,Cartwright:1994ly}, a
language with exceptions, and we are interested in seeing whether our
model might be fully abstract in that setting.

\section{Conclusions and Future Work}
\label{sec:conclusion}

In this paper we present an elementary semantics for a CBV
$\lambda$-calculus that represents a $\lambda$ abstraction with an
infinite set of finite tables.  We give a mechanized proof that this
semantics is correct with respect to the operational semantics of the
CBV $\lambda$-calculus and we present two case studies that begin to
demonstrate that this semantics is useful.  We leverage the
compositionality of our semantics in a proof of correctness for a
compiler optimization. We extend the semantics to handle parametric
polymorphism and prove type soundness, i.e., well-typed programs
cannot go wrong.

Of course, we have just scratched the surface in investigating how
well elementary semantics scales to full programming languages.  We
invite the reader to help us explore elementary semantics for mutable
state, exceptions, continuations, recursive types, dependent types,
objects, threads, shared memory, and low-level languages, to name just
a few. Regarding applications, there is plenty to try regarding proofs
of program correctness and compiler correctness. For your next
programming languages project, give elementary semantics a try!

{\tiny
\bibliographystyle{plainnat}
\bibliography{all}

\begin{thebibliography}{73}
\providecommand{\natexlab}[1]{#1}
\providecommand{\url}[1]{\texttt{#1}}
\expandafter\ifx\csname urlstyle\endcsname\relax
  \providecommand{\doi}[1]{doi: #1}\else
  \providecommand{\doi}{doi: \begingroup \urlstyle{rm}\Url}\fi

\bibitem[Abramsky(1990)]{Abramsky:1990vv}
S.~Abramsky.
\newblock The lazy lambda calculus.
\newblock In \emph{Research Topics in Functional Programming}, pages 65--116.
  Addison-Welsey, Reading, MA, 1990.
\newblock ISBN 0-201-17236-4.

\bibitem[Abramsky(1987)]{Abramsky:1987aa}
Samson Abramsky.
\newblock \emph{Domain Theory and the Logic of Observable Properties}.
\newblock PhD thesis, University of London, October 1987.

\bibitem[Abramsky et~al.(2000)Abramsky, Jagadeesan, and
  Malacaria]{Abramsky:2000kx}
Samson Abramsky, Radha Jagadeesan, and Pasquale Malacaria.
\newblock Full abstraction for pcf.
\newblock \emph{Inf. Comput.}, 163\penalty0 (2):\penalty0 409--470, 2000.
\newblock ISSN 0890-5401.

\bibitem[Aczel(1988)]{Aczel:1988aa}
Peter Aczel.
\newblock \emph{Non-well-founded Sets}.
\newblock Number~14 in CSLI Lecture Notes. Center for the Study of Language and
  Information, 1988.

\bibitem[Ahmed et~al.(2009)Ahmed, Dreyer, and Rossberg]{Ahmed:2009aa}
Amal Ahmed, Derek Dreyer, and Andreas Rossberg.
\newblock State-dependent representation independence.
\newblock In \emph{Proceedings of the 36th Annual ACM SIGPLAN-SIGACT Symposium
  on Principles of Programming Languages}, POPL '09, pages 340--353, New York,
  NY, USA, 2009. ACM.
\newblock ISBN 978-1-60558-379-2.
\newblock \doi{10.1145/1480881.1480925}.
\newblock URL \url{http://doi.acm.org/10.1145/1480881.1480925}.

\bibitem[Alessi et~al.(2003)Alessi, Barbanera, and
  Dezani-Ciancaglini]{Alessi:2003aa}
Fabio Alessi, Franco Barbanera, and Mariangiola Dezani-Ciancaglini.
\newblock Intersection types and computational rules.
\newblock \emph{Electronic Notes in Theoretical Computer Science}, 84:\penalty0
  45 -- 59, 2003.
\newblock ISSN 1571-0661.
\newblock \doi{http://dx.doi.org/10.1016/S1571-0661(04)80843-0}.
\newblock URL
  \url{http://www.sciencedirect.com/science/article/pii/S1571066104808430}.
\newblock WoLLIC'2003, 10th Workshop on Logic, Language, Information and
  Computation.

\bibitem[Alessi et~al.(2004)Alessi, Barbanera, and
  Dezani-Ciancaglini]{Alessi:2004aa}
Fabio Alessi, Franco Barbanera, and Mariangiola Dezani-Ciancaglini.
\newblock \emph{Tailoring Filter Models}, pages 17--33.
\newblock Springer Berlin Heidelberg, Berlin, Heidelberg, 2004.
\newblock ISBN 978-3-540-24849-1.
\newblock \doi{10.1007/978-3-540-24849-1_2}.
\newblock URL \url{http://dx.doi.org/10.1007/978-3-540-24849-1_2}.

\bibitem[Alessi et~al.(2006)Alessi, Barbanera, and
  Dezani-Ciancaglini]{Alessi:2006aa}
Fabio Alessi, Franco Barbanera, and Mariangiola Dezani-Ciancaglini.
\newblock Intersection types and lambda models.
\newblock \emph{Theoretical Compututer Science}, 355\penalty0 (2):\penalty0
  108--126, 2006.

\bibitem[Amin and Rompf(2017)]{Amin:2017aa}
Nada Amin and Tiark Rompf.
\newblock Type soundness proofs with definitional interpreters.
\newblock In \emph{Proceedings of the 44th ACM SIGPLAN Symposium on Principles
  of Programming Languages}, POPL 2017, pages 666--679, New York, NY, USA,
  2017. ACM.
\newblock ISBN 978-1-4503-4660-3.
\newblock \doi{10.1145/3009837.3009866}.
\newblock URL \url{http://doi.acm.org/10.1145/3009837.3009866}.

\bibitem[Banerjee(1997)]{Banerjee:1997aa}
Anindya Banerjee.
\newblock A modular, polyvariant and type-based closure analysis.
\newblock In \emph{Proceedings of the Second ACM SIGPLAN International
  Conference on Functional Programming}, ICFP '97, pages 1--10, New York, NY,
  USA, 1997. ACM.
\newblock ISBN 0-89791-918-1.
\newblock \doi{10.1145/258948.258951}.
\newblock URL \url{http://doi.acm.org/10.1145/258948.258951}.

\bibitem[Barendregt et~al.(2013)Barendregt, Dekkers, and
  Statman]{Barendregt:2013aa}
Henk Barendregt, Wil Dekkers, and Richard Statman.
\newblock \emph{Lambda Calculus with Types}.
\newblock Perspectives in Logic. Cambridge University Press, 2013.
\newblock \doi{10.1017/CBO9781139032636}.

\bibitem[Benton et~al.(2009)Benton, Kennedy, and Varming]{Benton:2009ab}
Nick Benton, Andrew Kennedy, and Carsten Varming.
\newblock \emph{Some Domain Theory and Denotational Semantics in Coq}, pages
  115--130.
\newblock Springer Berlin Heidelberg, Berlin, Heidelberg, 2009.
\newblock ISBN 978-3-642-03359-9.
\newblock \doi{10.1007/978-3-642-03359-9_10}.
\newblock URL \url{http://dx.doi.org/10.1007/978-3-642-03359-9_10}.

\bibitem[Cartwright and Felleisen(1992)]{Cartwright:1992ys}
Robert Cartwright and Matthias Felleisen.
\newblock Observable sequentiality and full abstraction.
\newblock In \emph{POPL '92: Proceedings of the 19th ACM SIGPLAN-SIGACT
  symposium on Principles of programming languages}, pages 328--342, New York,
  NY, USA, 1992. ACM Press.
\newblock ISBN 0-89791-453-8.

\bibitem[Cartwright et~al.(1994)Cartwright, Curien, and
  Felleisen]{Cartwright:1994ly}
Robert Cartwright, Pierre-Louis Curien, and Matthias Felleisen.
\newblock Fully abstract semantics for observably sequential languages.
\newblock \emph{Inf. Comput.}, 111\penalty0 (2):\penalty0 297--401, 1994.
\newblock ISSN 0890-5401.

\bibitem[Chargu{\'e}raud(2013)]{Chargueraud:2013aa}
Arthur Chargu{\'e}raud.
\newblock Pretty-big-step semantics.
\newblock In \emph{Proceedings of the 22Nd European Conference on Programming
  Languages and Systems}, ESOP'13, pages 41--60, Berlin, Heidelberg, 2013.
  Springer-Verlag.
\newblock ISBN 978-3-642-37035-9.
\newblock \doi{10.1007/978-3-642-37036-6_3}.
\newblock URL \url{http://dx.doi.org/10.1007/978-3-642-37036-6_3}.

\bibitem[Church(1932)]{Church:1932aa}
Alonzo Church.
\newblock A set of postulates for the foundation of logic.
\newblock \emph{Annals of Mathematics}, 33\penalty0 (2):\penalty0 pp. 346--366,
  1932.
\newblock ISSN 0003486X.
\newblock URL \url{http://www.jstor.org/stable/1968337}.

\bibitem[Coppo et~al.(1979)Coppo, Dezani-Ciancaglini, and Salle']{Coppo:1979aa}
M.~Coppo, M.~Dezani-Ciancaglini, and P.~Salle'.
\newblock \emph{Functional characterization of some semantic equalities inside
  $\lambda$-calculus}, pages 133--146.
\newblock Springer Berlin Heidelberg, Berlin, Heidelberg, 1979.
\newblock ISBN 978-3-540-35168-9.
\newblock \doi{10.1007/3-540-09510-1_11}.
\newblock URL \url{http://dx.doi.org/10.1007/3-540-09510-1_11}.

\bibitem[Coppo et~al.(1984)Coppo, Dezani-Ciancaglini, Honsell, and
  Longo]{Coppo:1984aa}
M.~Coppo, M.~Dezani-Ciancaglini, F.~Honsell, and G.~Longo.
\newblock Extended type structures and filter lambda models.
\newblock In G.~Longo G.~Lolli and A.~Marcja, editors, \emph{Logic Colloquium
  '82}, volume 112 of \emph{Studies in Logic and the Foundations of
  Mathematics}, pages 241 -- 262. Elsevier, 1984.
\newblock \doi{http://dx.doi.org/10.1016/S0049-237X(08)71819-6}.
\newblock URL
  \url{http://www.sciencedirect.com/science/article/pii/S0049237X08718196}.

\bibitem[Coppo et~al.(1987)Coppo, Dezani-Ciancaglini, and Zacchi]{Coppo:1987aa}
M.~Coppo, M.~Dezani-Ciancaglini, and M.~Zacchi.
\newblock Type theories, normal forms, and d∞-lambda-models.
\newblock \emph{Information and Computation}, 72\penalty0 (2):\penalty0 85 --
  116, 1987.
\newblock ISSN 0890-5401.
\newblock \doi{http://dx.doi.org/10.1016/0890-5401(87)90042-3}.
\newblock URL
  \url{http://www.sciencedirect.com/science/article/pii/0890540187900423}.

\bibitem[Dockins(2014)]{Dockins:2014aa}
Robert Dockins.
\newblock Formalized, effective domain theory in coq.
\newblock In \emph{Interactive Theorem Proving}, ITP, 2014.

\bibitem[Dudenhefner and Rehof(2017)]{Dudenhefner:2017aa}
Andrej Dudenhefner and Jakob Rehof.
\newblock Intersection type calculi of bounded dimension.
\newblock In \emph{Proceedings of the 44th ACM SIGPLAN Symposium on Principles
  of Programming Languages}, POPL 2017, pages 653--665, New York, NY, USA,
  2017. ACM.
\newblock ISBN 978-1-4503-4660-3.
\newblock \doi{10.1145/3009837.3009862}.
\newblock URL \url{http://doi.acm.org/10.1145/3009837.3009862}.

\bibitem[Egidi et~al.(1992)Egidi, Honsell, and Della~Rocca]{Egidi:1992aa}
Lavinia Egidi, Furio Honsell, and Simona~Ronchi Della~Rocca.
\newblock Operational, denotational and logical descriptions: A case study.
\newblock \emph{Fundam. Inf.}, 16\penalty0 (2):\penalty0 149--169, February
  1992.
\newblock ISSN 0169-2968.
\newblock URL \url{http://dl.acm.org/citation.cfm?id=161643.161646}.

\bibitem[Engeler(1981)]{Engeler:1981aa}
Erwin Engeler.
\newblock Algebras and combinators.
\newblock \emph{algebra universalis}, 13\penalty0 (1):\penalty0 389--392, Dec
  1981.
\newblock ISSN 1420-8911.
\newblock \doi{10.1007/BF02483849}.
\newblock URL \url{https://doi.org/10.1007/BF02483849}.

\bibitem[Felleisen et~al.(1987)Felleisen, Friedman, Kohlbecker, and
  Duba]{Felleisen:1987ab}
Matthias Felleisen, Daniel~P. Friedman, Eugene Kohlbecker, and Bruce Duba.
\newblock A syntactic theory of sequential control.
\newblock \emph{Theoretical Computer Science}, 52\penalty0 (3):\penalty0 205 --
  237, 1987.
\newblock ISSN 0304-3975.
\newblock \doi{http://dx.doi.org/10.1016/0304-3975(87)90109-5}.
\newblock URL
  \url{http://www.sciencedirect.com/science/article/pii/0304397587901095}.

\bibitem[Floyd(1969)]{Floyd:1969aa}
R.~W. Floyd.
\newblock Assigning meanings to programs.
\newblock In \emph{Proceedings of Symposia in Applied Mathematics}, volume~19,
  pages 19--32, 1969.

\bibitem[Frankle et~al.(2016)Frankle, Osera, Walker, and
  Zdancewic]{Frankle:2016aa}
Jonathan Frankle, Peter-Michael Osera, David Walker, and Steve Zdancewic.
\newblock Example-directed synthesis: A type-theoretic interpretation.
\newblock In \emph{Proceedings of the 43rd Annual ACM SIGPLAN-SIGACT Symposium
  on Principles of Programming Languages}, POPL '16, pages 802--815, New York,
  NY, USA, 2016. ACM.
\newblock ISBN 978-1-4503-3549-2.
\newblock \doi{10.1145/2837614.2837629}.
\newblock URL \url{http://doi.acm.org/10.1145/2837614.2837629}.

\bibitem[Frisch et~al.(2008)Frisch, Castagna, and Benzaken]{Frisch:2008aa}
Alain Frisch, Giuseppe Castagna, and V{\'e}ronique Benzaken.
\newblock Semantic subtyping: Dealing set-theoretically with function, union,
  intersection, and negation types.
\newblock \emph{J. ACM}, 55\penalty0 (4):\penalty0 19:1--19:64, September 2008.

\bibitem[Girard(1972)]{GIRARD72}
Jean-Yves Girard.
\newblock \emph{Interpretation fonctionelle et elimination des coupures de
  l'arithmetique d'ordre superieur}.
\newblock PhD thesis, Paris, France, 1972.

\bibitem[Gunter(1992)]{Gunter:1992aa}
Carl~A. Gunter.
\newblock \emph{Semantics of Programming Languages: Structures and Techniques}.
\newblock MIT Press, Cambridge, MA, USA, 1992.
\newblock ISBN 0-262-07143-6.

\bibitem[Harper(2012)]{Harper:2012aa}
Professor~Robert Harper.
\newblock \emph{Practical Foundations for Programming Languages}.
\newblock Cambridge University Press, New York, NY, USA, 2012.
\newblock ISBN 1107029570, 9781107029576.

\bibitem[Hindley(1982)]{Hindley:1982aa}
J.~R. Hindley.
\newblock \emph{The simple semantics for Coppo-Dezani-Sall{\'e} types}, pages
  212--226.
\newblock Springer Berlin Heidelberg, Berlin, Heidelberg, 1982.
\newblock ISBN 978-3-540-39184-5.
\newblock \doi{10.1007/3-540-11494-7_15}.
\newblock URL \url{http://dx.doi.org/10.1007/3-540-11494-7_15}.

\bibitem[Hoare(1969)]{Hoare:1969kw}
C.~A.~R. Hoare.
\newblock An axiomatic basis for computer programming.
\newblock \emph{Commun. ACM}, 12\penalty0 (10):\penalty0 576--580, 1969.
\newblock ISSN 0001-0782.

\bibitem[Hur and Dreyer(2011)]{Hur:2011aa}
Chung-Kil Hur and Derek Dreyer.
\newblock A kripke logical relation between ml and assembly.
\newblock In \emph{Proceedings of the 38th Annual ACM SIGPLAN-SIGACT Symposium
  on Principles of Programming Languages}, POPL '11, pages 133--146, New York,
  NY, USA, 2011. ACM.
\newblock ISBN 978-1-4503-0490-0.
\newblock \doi{10.1145/1926385.1926402}.
\newblock URL \url{http://doi.acm.org/10.1145/1926385.1926402}.

\bibitem[Hyland and Ong(2000)]{Hyland:2000ve}
J.~M.~E. Hyland and C.-H.~L. Ong.
\newblock On full abstraction for pcf: I, ii, and iii.
\newblock \emph{Inf. Comput.}, 163\penalty0 (2):\penalty0 285--408, 2000.
\newblock ISSN 0890-5401.

\bibitem[Jagannathan and Wright(1996)]{Jagannathan:1996aa}
Suresh Jagannathan and Andrew Wright.
\newblock Flow-directed inlining.
\newblock In \emph{Proceedings of the ACM SIGPLAN 1996 Conference on
  Programming Language Design and Implementation}, PLDI '96, pages 193--205,
  New York, NY, USA, 1996. ACM.
\newblock ISBN 0-89791-795-2.
\newblock \doi{10.1145/231379.231417}.
\newblock URL \url{http://doi.acm.org/10.1145/231379.231417}.

\bibitem[Jeffrey and Rathke(2005)]{Jeffrey:2005kx}
A.~S.~A. Jeffrey and J.~Rathke.
\newblock Java jr.: Fully abstract trace semantics for a core java language.
\newblock In \emph{Proc. European Symposium on Programming}, 2005.

\bibitem[Jeffrey(1993)]{Jeffrey:1993aa}
Alan Jeffrey.
\newblock A fully abstract semantics for concurrent graph reduction.
\newblock Technical Report 12/93, University of Sussex, 1993.

\bibitem[Kahn(1987)]{Kahn:1987aa}
Gilles Kahn.
\newblock Natural semantics.
\newblock In \emph{Symposium on Theoretical Aspects of Computer Science}, pages
  22--39, 1987.

\bibitem[Landin(1964)]{Landin:1964dk}
P.~J. Landin.
\newblock The mechanical evaluation of expressions.
\newblock \emph{The Computer Journal}, 6\penalty0 (4):\penalty0 308--320, 1964.

\bibitem[Leroy(2009)]{Leroy:2009aa}
Xavier Leroy.
\newblock Formal verification of a realistic compiler.
\newblock \emph{Commun. ACM}, 52\penalty0 (7):\penalty0 107--115, July 2009.
\newblock ISSN 0001-0782.
\newblock \doi{10.1145/1538788.1538814}.
\newblock URL \url{http://doi.acm.org/10.1145/1538788.1538814}.

\bibitem[Martin-L\"{o}f(1985)]{Martin-Lof:1985aa}
P.~Martin-L\"{o}f.
\newblock Constructive mathematics and computer programming.
\newblock In \emph{Proc. Of a Discussion Meeting of the Royal Society of London
  on Mathematical Logic and Programming Languages}, pages 167--184, Upper
  Saddle River, NJ, USA, 1985. Prentice-Hall, Inc.
\newblock ISBN 0-13-561465-1.
\newblock URL \url{http://dl.acm.org/citation.cfm?id=3721.3731}.

\bibitem[McCarthy(1960)]{McCarthy:1960dz}
John McCarthy.
\newblock Recursive functions of symbolic expressions and their computation by
  machine, part i.
\newblock \emph{Commun. ACM}, 3\penalty0 (4):\penalty0 184--195, 1960.
\newblock ISSN 0001-0782.

\bibitem[Milner(1978)]{Milner:1978kh}
Robin Milner.
\newblock A theory of type polymorphism in programming.
\newblock \emph{Journal of Computer and System Sciences}, 17\penalty0
  (3):\penalty0 348--375, 1978.

\bibitem[Milner(1995)]{Milner:1995aa}
Robin Milner.
\newblock \emph{Communication and Concurrency}.
\newblock Prentice Hall International (UK) Ltd., Hertfordshire, UK, UK, 1995.
\newblock ISBN 0-13-115007-3.

\bibitem[Morris(1968)]{James-Hiram-Morris:1968kx}
James~H. Morris.
\newblock \emph{Lambda-calculus Models of Programming Languages}.
\newblock PhD thesis, MIT, Cambridge, MA, USA, December 1968.

\bibitem[Owens et~al.(2016)Owens, Myreen, Kumar, and Tan]{Owens:2016aa}
Scott Owens, Magnus~O. Myreen, Ramana Kumar, and Yong~Kiam Tan.
\newblock Functional big-step semantics.
\newblock In \emph{Proceedings of the 25th European Symposium on Programming
  Languages and Systems - Volume 9632}, pages 589--615, New York, NY, USA,
  2016. Springer-Verlag New York, Inc.
\newblock ISBN 978-3-662-49497-4.
\newblock \doi{10.1007/978-3-662-49498-1_23}.
\newblock URL \url{http://dx.doi.org/10.1007/978-3-662-49498-1_23}.

\bibitem[Pierce(2002)]{Pierce:2002hj}
Benjamin~C. Pierce.
\newblock \emph{Types and {P}rogramming {L}anguages}.
\newblock MIT Press, 2002.

\bibitem[Plotkin(1973)]{Plotkin:1973fk}
G.D. Plotkin.
\newblock Lambda-definability and logical relations.
\newblock Technical Report SAI-RM-4, University of Edinburgh, School of
  Artificial Intelligence, October 1973.

\bibitem[Plotkin(1981)]{Plotkin:1981aa}
G.D. Plotkin.
\newblock A structural approach to operational semantics.
\newblock Technical Report DAIMI FN-19, Computer Science Dept., Aarhus
  Universitet, 1981.
\newblock URL \url{https://books.google.com/books?id=tt7sjgEACAAJ}.

\bibitem[Plotkin(1972)]{Plotkin:1972aa}
Gordon~D. Plotkin.
\newblock A set-theoretical definition of application.
\newblock Technical Report MIP-R-95, University of Edinburgh, 1972.

\bibitem[Plotkin(1993)]{Plotkin:1993ab}
Gordon~D. Plotkin.
\newblock Set-theoretical and other elementary models of the λ-calculus.
\newblock \emph{Theoretical Computer Science}, 121\penalty0 (1):\penalty0 351
  -- 409, 1993.
\newblock ISSN 0304-3975.
\newblock \doi{http://dx.doi.org/10.1016/0304-3975(93)90094-A}.
\newblock URL
  \url{http://www.sciencedirect.com/science/article/pii/030439759390094A}.

\bibitem[Reynolds(1974)]{REYNOLDS74C}
John~C. Reynolds.
\newblock Towards a theory of type structure.
\newblock In \emph{Programming Symposium}, volume~19 of \emph{LNCS}, pages
  408--425. Springer-Verlag, 1974.

\bibitem[Reynolds(2004)]{Reynolds:2004fk}
John~C. Reynolds.
\newblock Toward a grainless semantics for shared-variable concurrency.
\newblock In \emph{Proceedings of the 24th Conference on Foundations of
  Software Technology and Theoretical Computer Science (FSTTCS 2004)}, December
  2004.

\bibitem[{Ronchi Della Rocca} and Paolini(2004)]{Rocca:2004aa}
Simona {Ronchi Della Rocca} and Luca Paolini.
\newblock \emph{The Parametric Lambda Calculus}.
\newblock Springer, 2004.

\bibitem[Schmidt(2003)]{Schmidt:2003aa}
David~A. Schmidt.
\newblock Programming language semantics.
\newblock In \emph{Encyclopedia of Computer Science}, pages 1463--1466. John
  Wiley and Sons Ltd., Chichester, UK, 2003.
\newblock ISBN 0-470-86412-5.
\newblock URL \url{http://dl.acm.org/citation.cfm?id=1074100.1074733}.

\bibitem[Scott(1970)]{Scott:1970dp}
Dana Scott.
\newblock Outline of a mathematical theory of computation.
\newblock Technical Report PRG-2, Oxford University, November 1970.

\bibitem[Scott(1976)]{Scott:1976lq}
Dana Scott.
\newblock Data types as lattices.
\newblock \emph{SIAM Journal on Computing}, 5\penalty0 (3):\penalty0 522--587,
  1976.

\bibitem[Scott and Strachey(1971)]{Scott:1971ab}
Dana Scott and Christopher Strachey.
\newblock Toward a mathematical semantics for computer languages.
\newblock Technical Report PRG-6, Oxford Programming Research Group, 1971.

\bibitem[Scott(1993)]{Scott:1993uq}
Dana~S. Scott.
\newblock A type-theoretic alternative to iswim, cuch, owhy.
\newblock \emph{Theoretical Computer Science}, 121\penalty0 (411-440), 1993.

\bibitem[Siek(2017)]{Siek:2017aa}
Jeremy Siek.
\newblock Declarative semantics for functional languages.
\newblock \emph{Archive of Formal Proofs}, July 2017.
\newblock ISSN 2150-914x.
\newblock \url{http://isa-afp.org/entries/Decl_Sem_Fun_PL.html}, Formal proof
  development.

\bibitem[Siek(2012)]{Siek:2012ac}
Jeremy~G. Siek.
\newblock Big-step, diverging or stuck?, July 2012.
\newblock URL
  \url{http://siek.blogspot.com/2012/07/big-step-diverging-or-stuck.html}.

\bibitem[Skorstengaard(2015)]{Skorstengaard:2015aa}
Lau Skorstengaard.
\newblock An introduction to logical relations.
\newblock Technical report, Aarhus University, 2015.

\bibitem[Smyth and Plotkin(1982)]{Smyth:1982aa}
M.~B. Smyth and Gordon~D. Plotkin.
\newblock The category-theoretic solution of recursive domain equations.
\newblock \emph{SIAM Journal on Computing}, 11\penalty0 (4), November 1982.

\bibitem[Statman(1985)]{Statman:1985aa}
R.~Statman.
\newblock Logical relations and the typed λ-calculus.
\newblock \emph{Information and Control}, 65\penalty0 (2--3):\penalty0 85 --
  97, 1985.
\newblock ISSN 0019-9958.
\newblock \doi{http://dx.doi.org/10.1016/S0019-9958(85)80001-2}.
\newblock URL
  \url{http://www.sciencedirect.com/science/article/pii/S0019995885800012}.

\bibitem[Tristan and Leroy(2008)]{Tristan:2008aa}
Jean-Baptiste Tristan and Xavier Leroy.
\newblock Formal verification of translation validators: A case study on
  instruction scheduling optimizations.
\newblock In \emph{Proceedings of the 35th Annual ACM SIGPLAN-SIGACT Symposium
  on Principles of Programming Languages}, POPL '08, pages 17--27, New York,
  NY, USA, 2008. ACM.
\newblock ISBN 978-1-59593-689-9.
\newblock \doi{10.1145/1328438.1328444}.
\newblock URL \url{http://doi.acm.org/10.1145/1328438.1328444}.

\bibitem[Turbak et~al.(1997)Turbak, Dimock, Muller, and Wells]{Turbak:1997aa}
Franklyn Turbak, Allyn Dimock, Robert Muller, and J.~B. Wells.
\newblock Compiling with polymorphic and polyvariant flow types.
\newblock In \emph{Workshop on Types in Compilation}, 1997.

\bibitem[Waddell and Dybvig(1997)]{Waddell:1997fk}
Oscar Waddell and R.~Kent Dybvig.
\newblock Fast and effective procedure inlining.
\newblock In \emph{Proceedings of the 4th International Symposium on Static
  Analysis}, SAS '97, pages 35--52, London, UK, 1997. Springer-Verlag.

\bibitem[Wadler(1989)]{Wadler:1989fk}
Philip Wadler.
\newblock Theorems for free!
\newblock In \emph{FPCA '89: Proceedings of the fourth international conference
  on Functional programming languages and computer architecture}, pages
  347--359. ACM, 1989.
\newblock ISBN 0-89791-328-0.

\bibitem[Wadler(1992)]{wadler92:_essence}
Philip Wadler.
\newblock The essence of functional programming.
\newblock In \emph{Symposium on {P}rinciples of {P}rogramming {L}anguages},
  1992.

\bibitem[Wand(1979)]{Wand:1979aa}
Mitchell Wand.
\newblock Fixed-point constructions in order-enriched categories.
\newblock \emph{Theoretical Computer Science}, 8\penalty0 (1):\penalty0 13 --
  30, 1979.
\newblock ISSN 0304-3975.
\newblock \doi{http://dx.doi.org/10.1016/0304-3975(79)90053-7}.
\newblock URL
  \url{http://www.sciencedirect.com/science/article/pii/0304397579900537}.

\bibitem[Winskel(1987)]{Winskel:1987fp}
Glynn Winskel.
\newblock Event structures.
\newblock In \emph{Advances in Petri Nets 1986}, number 255 in LNCS, 1987.

\bibitem[Winskel(1993)]{Winskel:1993uq}
Glynn Winskel.
\newblock \emph{{The Formal Semantics of Programming Languages}}.
\newblock {Foundations of Computing}. MIT Press, 1993.

\bibitem[Wright and Felleisen(1994)]{wright94:_type_soundness}
Andrew~K. Wright and Matthias Felleisen.
\newblock A syntactic approach to type soundness.
\newblock \emph{Information and Computation}, 115\penalty0 (1):\penalty0
  38--94, 1994.
\newblock ISSN 0890-5401.

\end{thebibliography}
}

\end{document}